\documentclass[a4paper]{article}

\usepackage{t1enc}
\usepackage{graphicx}
\usepackage{parskip}
\usepackage{amsmath}
\usepackage{mathtools}
\usepackage{amssymb}
\usepackage{color}
\usepackage{tcolorbox}
\usepackage{url}
\usepackage{amsthm}
\usepackage{newtxmath}
\usepackage{thmtools,thm-restate}
\usepackage{subcaption}
\usepackage{a4wide}

\usepackage{float}

\newtheorem{thm}{Theorem}[section]
\newtheorem{prp}[thm]{Proposition}
\newtheorem{lemma}[thm]{Lemma}

\theoremstyle{definition}
\newtheorem{remark}[thm]{Remark}

\numberwithin{equation}{section}

\usepackage[
    backend=biber,
    style=numeric,
    defernumbers=false,
  ]{biblatex}
\usepackage{eso-pic}
\usepackage{enumerate}
\usepackage{authblk}

\date{\small \today}

\begin{document}

\title{Sample Average Approximation for Portfolio Optimization under CVaR constraint in an (re)insurance context}

\author{Jérôme Lelong \thanks{Univ. Grenoble Alpes, CNRS, Grenoble INP, LJK, 38000 Grenoble, France. \texttt{jerome.lelong@univ-grenoble-alpes.fr}} \ \ Véronique Maume-Deschamps \thanks{Universite Claude Bernard Lyon 1, CNRS, Ecole Centrale de Lyon, INSA Lyon, Université Jean Monnet, ICJ UMR5208,
69622 Villeurbanne, France. \texttt{veronique.maume-deschamps@univ-lyon1.fr}} \\  William Thevenot \thanks{Universite Claude Bernard Lyon 1, CNRS, Ecole Centrale de Lyon, INSA Lyon, Université Jean Monnet, ICJ UMR5208,
        69622 Villeurbanne, France. and Risk Knowledge team at SCOR SE, Paris, France \texttt{thevenot@math.univ-lyon1.fr}} }

\maketitle

\begin{abstract}
    We consider optimal allocation problems with Conditional Value-At-Risk (CVaR) constraint. We prove, under very mild assumptions, the convergence of the Sample Average Approximation method (SAA) applied to this problem, and we also exhibit a convergence rate and discuss the uniqueness of the solution. These results give (re)insurers a practical solution to portfolio optimization under market regulatory constraints, i.e. a certain level of risk.
\end{abstract}


\section{Introduction}

It is in the interest of (re)insurance companies to reduce risk through diversification in order to enhance their risk–return profile. The approach developed in this paper provides a quantitative framework to identify an ideal mix of risk exposures that balances expected profit with the Solvency Capital Requirement (SCR), which is represented in our model by the Conditional Value-at-Risk (CVaR). In the European market, insurers are subject to Solvency II regulations, which require that they hold an amount of eligible own funds at least equal to the SCR. This capital buffer is defined as the amount needed to ensure that the (re)insurance company can meet its obligations over the next 12 months with a probability of at least 99.5\%. While the SCR is formally based on the Value-at-Risk (VaR) at the confidence level $\alpha = 0.995$, the Conditional Value-at-Risk (CVaR) is commonly used as a more tractable proxy. Accordingly, we adopt CVaR as the risk measure in our analysis.

The classic approach to portfolio optimization was introduced by Markowitz in 1952 \cite{markowitz1952modern}. It consists in the maximization of the expectation under the constraint of maximum variance or, equivalently, minimizing the variance of the portfolio, for a fixed return, this problem is called the mean-variance optimization. Its counterpart formulation for the conditional value-at-risk (CVaR) is the mean-CVaR optimization.

We model the (re)insurance asset market with business lines represented by the random vector $\textbf{X}$ of asset returns, taking values in a subset $\mathcal{R}_\textbf{X}$ of $\mathbb{R}^d$, unlike in the classical financial case $\textbf{X}$ can take negative values, as they represent the results of the business lines. We assume that $\mathbb{E}(|\textbf{X}|) < +\infty$. A (re)insurance portfolio is defined by a vector $\boldsymbol{\gamma} \in \mathbb{R}^d$ representing the quantity held in each business line by the (re)insurer. In our business case, it cannot be interpreted in the same way as in the classical finance formulation, there is no initial wealth and it does not represent the repartition of it. $\gamma$ represents the (re)insurance company activities.

Let us fix some notations, with $\alpha \in ]0,1[$:
\begin{equation*}
    \begin{aligned}
         & V_\alpha(\boldsymbol{\gamma}) = VaR_\alpha(-\boldsymbol{\gamma}^T \textbf{X}) = \min \left\{ M \in \mathbb{R} : \mathbb{P}\left( -\boldsymbol{\gamma}^T \textbf{X} \leq M \right) \geq \alpha \right\},                 \\
         & C_\alpha(\boldsymbol{\gamma}) = CVaR_\alpha(-\boldsymbol{\gamma}^T \textbf{X}) = \mathbb{E}\left( -\boldsymbol{\gamma}^T \textbf{X} \big| -\boldsymbol{\gamma}^T \textbf{X} \geq V_\alpha(\boldsymbol{\gamma}) \right). \\
    \end{aligned}
\end{equation*}

Our original goal is to solve the following equation with a fixed $\alpha \in ]0,1[$ and constraints on the weights and a capital requirement limit $K>0$. We consider a loss function  $L_0$ , which commonly depends on $ C_\alpha(\boldsymbol{\gamma})$. Even though $L_0(\boldsymbol{\gamma}, \mathbf{X})$ may already account for this dependency implicitly, it is convenient to make the dependence on $C_\alpha(\boldsymbol{\gamma})$ explicit for the purposes of future analysis.

\begin{equation}\label{eq:P}
    \begin{aligned}
        v^* := & \operatorname*{\inf}_{\boldsymbol{\gamma} \in \mathbb{R}_+^d} &  & \mathbb{E}(L_0(\boldsymbol{\gamma}, C_\alpha(\boldsymbol{\gamma}), \textbf{X}))                            \\
               & \textrm{s.t.}                                                 &  & \gamma_{i}^{low} \leq \gamma_i \leq \gamma_{i}^{up}                             & \forall i \in \{1,..,d\} \\
               & \textrm{s.t.}                                                 &  & C_\alpha(\boldsymbol{\gamma}) \leq K.                                                                      \\
    \end{aligned}
\end{equation}

Through reinsurance and other risk transfer solutions, (re)insurers can pilot their exposure to individual risks and improve overall diversification, thereby increasing the return per unit of risk. Achieving this objective requires identifying an ideal risk composition, both at a given moment and over time through underwriting and portfolio management. Reinsurers, in particular, steer their portfolio by adjusting the reinsurance shares they offer during tenders, according to their risk appetite and renewal strategy. This strategic allocation process, which does not rely on the internal structure of each business line, is formalized in our model through the minimization of a loss function, or equivalently the maximization of a return function, under the CVaR constraints that reflect regulatory capital requirements. The allocation vector $\boldsymbol{\gamma}$ may also implicitly account for market size, legacy positions, and underwriting partnerships, while supporting the design of multi-period strategies aimed at gradually reaching a target risk exposure.

The standard formulation of the CVaR is not convenient, as it requires computing the Value at Risk followed by the estimation of a conditional expectation. An alternative approach was introduced by R.T. Rockafellar and S. Uryasev in 2000~\cite{Uryasev2000:7}, and later extended by P. Krokhmal, J. Palmquist, and S. Uryasev in 2002~\cite{Uryasev2002:9}. They proposed an embedding technique that reformulates the CVaR as a minimization problem involving an auxiliary parameter, which can be integrated into the global optimization framework. This reformulation avoids the inconvenient structure of the original formulation and significantly reduces the associated computational burden.

In \cite{Uryasev2002:9}, the problem is addressed using linear programming; however, this method can become computationally expensive when dealing with large data samples. In this work, we opt for the \textit{Sample Average Approximation} (SAA) method, as introduced by Rubinstein and Shapiro~\cite{rubinstein1993discrete}.

For this formulation with explicit constraints, no convergence or convergence speed results with the SAA method has been published as far as we know, the closest result to our work is \cite{wang2008sample}. In this last one, the function to be minimized does not depend on the data sample. We also have two results closed to our work which are \cite{sun2017optimal} and \cite{yao2013mean} where a parametric method is used to find the data structure and then minimize with this parameterization which is interesting with little sample size, but which is not our case because we assume that the distribution of $\textbf{X}$ is known and can be simulated. Moreover, it's not exactly the same formulation, which implies different difficulties in solving the problem.

Under convexity, continuity, integrability assumptions, we prove a.s. the convergence and find a rate of convergence for the SAA version in the case where the function to be minimized depends on the data sample as do the constraint. Moreover, if the CVaR appears in the function to be minimized, we show that for the optimization, under monotonic assumption, it can be replaced by the auxiliary function introduced in \cite{Uryasev2002:9} and \cite{Uryasev2000:7}. We also propose a sufficient condition to obtain the uniqueness of the solution.

The paper is organized as follows. Section \ref{section:context_and_pb_presentation} describes the context and states the problem. Section \ref{section:sample_average_approx} introduces the SAA method and presents our main results on convergence (Theorem \ref{th:thm_final_convergence}), uniqueness (Theorems \ref{th:theorem_final_unicity} and \ref{th:thm_unicity}) and convergence rate (Theorem \ref{th:global_rate_cv}). Section \ref{section:numerical_study} is dedicated to some numerical studies on two different cases, where we show how these results can be used concretely in the context of a (re)insurer seeking to balance its various exposures.

\section{Optimization with CVaR constraint}\label{section:context_and_pb_presentation}
This section presents the transformation of the original problem and the SAA method applied to this equivalent problem.

\subsection{An equivalent setting}\label{subse:eq_setting}
Note that the Conditional Value-At-Risk can be rewritten using the Value-At-Risk as follows,
\begin{equation*}
    C_\alpha(\boldsymbol{\gamma}) = V_\alpha(\boldsymbol{\gamma}) + (1-\alpha)^{-1} \mathbb{E}\left( (-\boldsymbol{\gamma}^T \textbf{X}- V_\alpha(\boldsymbol{\gamma}))^+ \right).
\end{equation*}
We introduce $g(\boldsymbol{\gamma}, \zeta) = \zeta + (1-\alpha)^{-1} \mathbb{E}\left( (-\boldsymbol{\gamma}^T \textbf{X}- \zeta)^+ \right)$ for any $\zeta \in \mathbb{R}$.

Rockafellar, R. T.,  Uryasev, S. (2000) \cite{Uryasev2000:7} showed that $g(\boldsymbol{\gamma}, . )$ is convex, continuously differentiable and that for any $\boldsymbol{\gamma}$, $C_\alpha(\boldsymbol{\gamma})$ can be determined by minimising $g(\boldsymbol{\gamma},.)$:
\begin{align*}
     & C_\alpha(\boldsymbol{\gamma}) = \operatorname*{\min}_{\boldsymbol{\zeta \in \mathbb{R}}} g(\boldsymbol{\gamma}, \zeta)\text{,} \ \ A_\alpha(\boldsymbol{\gamma}) = \operatorname*{argmin}_{\boldsymbol{\zeta \in \mathbb{R}}} g(\boldsymbol{\gamma}, \zeta)\text{,} \\
     & V_\alpha(\boldsymbol{\gamma}) = \min(\{x\in   A_\alpha(\boldsymbol{\gamma})\}) \text{,} \ \  C_\alpha(\boldsymbol{\gamma}) = g(\boldsymbol{\gamma}, V_\alpha(\boldsymbol{\gamma})).
\end{align*}

Krokhmal P., Jonas Palmquist J., Uryasev S. (2002) \cite{Uryasev2002:9} showed the equivalence of the following two optimization problems in the sense that their objectives functions achieve the same minimum values. (Proof in Annex 1)

\begin{restatable}{prp}{equivalenceprop}
    \label{th:equivalence}
    For any function $R : \mathbb{R}^d \times \mathbb{R} \rightarrow \mathbb{R}$ such that for any $\boldsymbol{\gamma} \in \mathbb{R}^d$, $R(\boldsymbol{\gamma}, . )$ is not increasing then the following two problems are equivalent
    \begin{equation}\label{eq:Reward_True}
        \operatorname*{\inf}_{\boldsymbol{\gamma} \in \Xi} -R(\boldsymbol{\gamma}, C_\alpha(\boldsymbol{\gamma})) \ \ \text{s.t.} \ \ C_\alpha(\boldsymbol{\gamma})\leq K
    \end{equation}
    and
    \begin{equation}\label{eq:Reward_True_F}
        \operatorname*{\inf}_{(\boldsymbol{\gamma},\zeta) \in \Xi \times \mathbb{R}} -R(\boldsymbol{\gamma}, g(\boldsymbol{\gamma}, \zeta))  \ \ \text{s.t.} \ \ g(\boldsymbol{\gamma}, \zeta)\leq K.
    \end{equation}
    Moreover, if the $C_\alpha$ constraint is active in $\eqref{eq:Reward_True}$, $(\boldsymbol{\gamma}^*,\zeta^*)$ achieves the minimum of \eqref{eq:Reward_True_F} solution if and only if $\boldsymbol{\gamma}^*$ achieves the minimum of $\eqref{eq:Reward_True}$ and $\zeta^* \in A_\alpha(\boldsymbol{\gamma}^*)$.
\end{restatable}

It is convenient to write $L(\boldsymbol{\gamma},\zeta,\textbf{X}) = L_0(\boldsymbol{\gamma}, g(\boldsymbol{\gamma}, \zeta),\textbf{X})$ for future analysis.

Note that $g$ is also continuous and convex.

This result is fundamental in transforming our problem into a simpler one where it is no longer necessary to compute the value of the function $C_\alpha(\boldsymbol{\gamma})$, which is long and complex because we need to compute $V_\alpha(\boldsymbol{\gamma})$ first. Instead of that, we compute $g(\boldsymbol{\gamma}, \zeta)$ which is no longer the result of an optimisation problem. Using the equivalence between~\eqref{eq:Reward_True} and~\eqref{eq:Reward_True_F},~\eqref{eq:P} may be written as:

\begin{equation}\label{eq:P-eq}
    \begin{aligned}
        v^* := & \inf_{\boldsymbol{\gamma}, \zeta \in \mathbb{R}_+^d \times \mathbb{R}} &  & \mathbb{E}(L(\boldsymbol{\gamma},\zeta, \textbf{X}))             \\
               & \textrm{s.t.}                                                                         &  & \gamma_{i}^{low} \leq \gamma_i \leq \gamma_{i}^{up}  & \forall i \\
               & \textrm{s.t.}                                                                         &  & \zeta^{low} \leq \zeta \leq \zeta^{up}               & \forall i \\
               & \textrm{s.t.}                                                                         &  & g(\boldsymbol{\gamma}, \zeta) \leq K                             \\
    \end{aligned}
\end{equation}

where $\zeta^{low} \leq \min_{\gamma_{i}^{low} \leq \gamma_i \leq \gamma_{i}^{up} \ \forall i} VaR_\alpha (-\boldsymbol{\gamma}^T \textbf{X})$ and $ \zeta^{up} \geq \max_{\gamma_{i}^{low} \leq \gamma_i \leq \gamma_{i}^{up} \ \forall i} VaR_\alpha(-\boldsymbol{\gamma}^T \textbf{X})$.

\begin{remark}
    The most commonly used loss function is $L_0(\boldsymbol{\gamma},\textbf{X}) = -\boldsymbol{\gamma}^T \textbf{X}$. But in order to represent the cost of capital in insurance, it is usual to penalise the function with a term proportional to the capital requirement i.e. the CVaR. So the loss could become $L_0(\boldsymbol{\gamma},\textbf{X}) = -\boldsymbol{\gamma}^T X + c \cdot C_\alpha(\boldsymbol{\gamma}) = L_0(\boldsymbol{\gamma}, C_\alpha(\boldsymbol{\gamma}),\textbf{X})$ with $c \in ]0,1[$ usually $0.05$. In this last case, $L$ becomes $L(\boldsymbol{\gamma},\zeta,\textbf{X}) = -\boldsymbol{\gamma}^T \textbf{X} + c \cdot g(\boldsymbol{\gamma}, \zeta)$ in the equivalent Problem~\eqref{eq:P-eq}.
\end{remark}

\subsection{Sample Average Approximation}\label{section:sample_average_approx}
Usually we do not have access to $\mathbb{E}(L(\boldsymbol{\gamma},\zeta, \textbf{X}))$, so we use a real or simulated sample of $\textbf{X}$ (SAA: Sample Average Approximation \cite{shapiro2003monte}). Let $(\textbf{X}^{(j)})_{j \ge 1}$ be a sequence of iid random variables with the same distribution as $\textbf{X}$. We consider an approximation of Problem~\eqref{eq:P-eq}. For $N \ge 1$, consider
\begin{equation*}
    g_N(\boldsymbol{\gamma}, \zeta) :=  \zeta + \frac{(1-\alpha)^{-1}}{N} \sum_{j=1}^N (-\boldsymbol{\gamma}^T \textbf{X}^{(j)} - \zeta)^+ \ \
\end{equation*}
and
\begin{equation*}
    \ell_N(\boldsymbol{\gamma},\zeta) := \frac{1}{N} \sum_{j=1}^N L(\boldsymbol{\gamma},\zeta,\textbf{X}^{(j)}).
\end{equation*}
The sample average approximation of~\eqref{eq:P-eq} is given by
\begin{equation} \label{eq:P-eq-SAA}
    \begin{aligned}
        v_N := & \inf_{\boldsymbol{\gamma}, \zeta \in \mathbb{R}_+^d \times \mathbb{R}} &  & \ell_N(\boldsymbol{\gamma},\zeta)                                      \\
               & \textrm{s.t.}                                                                         &  & 0 \leq \gamma_{i}^{low} \leq \gamma_i \leq \gamma_{i}^{up} & \forall i \\
               &                                                                                       &  & \zeta^{low} \leq \zeta \leq \zeta^{up}
        \\                                                                       &  &  & g_N(\boldsymbol{\gamma}, \zeta) \leq K.                                \\
    \end{aligned}
\end{equation}

Let us introduce $\Gamma:= \left\{ \boldsymbol{\gamma} :  0 \leq \gamma_{i}^{low} \leq \gamma_i \leq \gamma_{i}^{up} \ \forall i \right\}$ and $\mathcal{U} := \Gamma \times \left\{ \zeta :\zeta^{low} \leq \zeta \leq \zeta^{up} \right\}$.

Define $\ell: (\boldsymbol{\gamma},\zeta) \mapsto \mathbb{E}(L(\boldsymbol{\gamma}, \zeta,\textbf{X}))$, we assume the following properties
\begin{itemize}
    \item[(P1)]$\mathbb{P}((\boldsymbol{\gamma},\zeta) \rightarrow L(\boldsymbol{\gamma}, \zeta,\textbf{X}) \ \text{is continuous and convex} ) = 1$.
    \item[(P2)] $\mathbb{E}\left(\underset{(\boldsymbol{\gamma},\zeta) \in \mathcal{U}}{\sup}|L(\boldsymbol{\gamma},\zeta,\textbf{X})|\right)<+\infty$.
\end{itemize}

The functions $\ell_N$ and $\ell$ are continuous and convex by (P1) and (P2). Note that $g_N$ and $g$ are convex and continuous.

Now we define the feasibility sets of the original problem and of its SAA counterpart; $\mathcal{U}^K := \left\{ (\boldsymbol{\gamma},\zeta) : g(\boldsymbol{\gamma},\zeta) \leq K \right\} \cap \mathcal{U}$ and $\mathcal{U}_N^K := \left\{ (\boldsymbol{\gamma},\zeta) : g_N(\boldsymbol{\gamma},\zeta) \leq K \right\} \cap \mathcal{U}$. We assume that $\mathcal{U}, \mathcal{U}^K$ have non-empty interiors.

We also define the solution sets of the original problem and of its SAA counterpart.
\begin{equation*}
    S :=  \operatorname*{argmin}\limits_{(\boldsymbol{\gamma}, \zeta) \in \mathcal{U}^K} \ell(\boldsymbol{\gamma},\zeta) \ \ \text{;} \ \ S_N := \operatorname*{argmin}\limits_{(\boldsymbol{\gamma}, \zeta) \in \mathcal{U}_N^K} \ell_N(\boldsymbol{\gamma},\zeta),
\end{equation*}

Our goal is to prove that $v_N$ converges a.s. to $v^*$ and $\mathbb{D}(S_N, S) \underset{N \rightarrow \infty}{\longrightarrow} 0 $ a.s. with $\mathbb{D}(A,B) := \adjustlimits\sup_{x\in A} \inf_{x'\in B} \lVert x-x' \rVert$.

\section{Main results}\label{section:main_results}
\subsection{Convergence of the SAA method}

We shall use two results from Shapiro, A., Dentcheva, D., Ruszczynski, A. (2009) \cite{shapiro2009:8}, in order to obtain a.s. convergence of our SAA Method (Theorem \ref{th:shapiro-5.5}/ Theorem \ref{th:thm_final_convergence}).

\begin{thm}[Theorem 7.48 p375 in \cite{shapiro2009:8}]~\label{th:shapiro-7.48}\par
    Consider a random function $H: \Xi \times \mathcal{R}_X \rightarrow \mathbb{R}$, where $\Xi$ is a non-empty compact subset of $\mathbb{R}^d$. Define $h(\boldsymbol{\xi}) := \mathbb{E}(H(\boldsymbol{\xi}, \textbf{X}))$ and $h_N(\boldsymbol{\xi}) := N^-1 \sum_{j=1}^N H(\boldsymbol{\xi}, \textbf{X}^{(j)})$.
    Suppose that:
    \begin{itemize}
        \item[(i)] $\mathbb{P}( \xi \rightarrow H(\xi,\textbf{X}) \ \text{is continuous}) = 1 $.
        \item[(ii)] $\mathbb{E}\left(\underset{\xi \in \Xi}{\sup} |H(\xi,\textbf{X})|\right)< +\infty $.
    \end{itemize}
    Then, the function $h$ is continuous on $\Xi$, and $(h_N)_{N \in \mathbb{N}^*}$ converges to $h$ a.s. uniformly on $\Xi$ i.e.
    $$
        \sup _{\boldsymbol{\xi} \in \Xi}\left|h_N(\boldsymbol{\xi})-h(\boldsymbol{\xi})\right| \rightarrow 0  \ \text{a.s.} \text{ as } N \rightarrow \infty.
    $$
\end{thm}

We use Theorem~\ref{th:shapiro-7.48} with the functions $L$ and $\ell_N$ on the compact set $\Gamma$ using (P1) and (P2) to obtain the following proposition.
\begin{prp}\label{th:prop_cvunif_l}
    If (P1) and (P2) hold, then, the function $\ell$ is continuous and the sequence $(\ell_N)_{N \in \mathbb{N}^*}$ converges to $\ell$ a.s. uniformly on $\Gamma$.
\end{prp}

And we use Theorem~\ref{th:shapiro-7.48} again with $G((\boldsymbol{\gamma}, \zeta),\textbf{X}) = \zeta +(1-\alpha)^{-1} (-\boldsymbol{\gamma}^T \textbf{X} - \zeta)^+$ and $g_N$ on the compact set $\mathcal{U}$ to obtain the following proposition.

\begin{prp}\label{th:prop_cvunif_g}
    The sequence $(g_N)_{N \in \mathbb{N}^*}$ converges to $g$ a.s. uniformly on $\mathcal{U}$.
\end{prp}

The set $\mathcal{U}$ is compact and $\mathcal{U}^K$ is non-empty by assumption, so neither is $S$. The function $\ell$ is continuous and convex, by (P1) and (P2), and $(\ell_N)_{N \in \mathbb{N}^*}$ converges to $\ell$ a.s. uniformly on $\mathcal{U}$ from Proposition~\ref{th:prop_cvunif_l} with (P1) and (P2). Thanks to these results, we can state a modified version of Theorem 5.5 p160 in \cite{shapiro2009:8} adapted to our framework.

\begin{thm}[Restatement of Theorem 5.5 p160 in \cite{shapiro2009:8}]\label{th:shapiro-5.5}\par
    \item If the following conditions hold:
    \begin{enumerate}
        \item[(a)] Let $(u_N)_N$ be a sequence taking values in $\mathcal{U}$ and such that $u_N$ converges $\Bar{u}$ a.s. If $u_N \in \mathcal{U}^K_N$ for all $N$, then $\Bar{u} \in \mathcal{U}^K$.
        \item[(b)] For $u \in S $, there exists a sequence $(u_N)_{N \in \mathbb{N}^*}$ such that $u_N \in \mathcal{U}_N^K$ for all $N$ and $u_N \underset{N \rightarrow \infty}{\longrightarrow} u $ a.s.
        \item[(c)] A.s. for $N$ large enough the set $S_N$ is nonempty and $S_N \subset \mathcal{U}$.
    \end{enumerate}
    Then $v_N \rightarrow v^*$ a.s. and $\mathbb{D}(S_N, S) \rightarrow 0 $ a.s. when $N\rightarrow \infty$.
\end{thm}

\begin{prp}\label{th:prop_cv_a}
    Let $(u_N)_N$ be a sequence taking values in $\mathcal{U}$ and such that $u_N$ converges $\Bar{u}$ a.s. If $u_N \in \mathcal{U}^K_N$ for all $N$, then $\Bar{u} \in \mathcal{U}^K$.
\end{prp}
\begin{proof}
    Let $(u_N)_{N\in\mathbb{N}^*}$ be a sequence taking values in $\mathcal{U}$ such that $u_N$ converges $\Bar{u}$ a.s. and such that $u_N \in \mathcal{U}_N^K$ for all $N$. We have
    \begin{equation}\label{eq:in_gamma_N}
        g_N(u_N) \leq K.
    \end{equation}

    By Proposition~\ref{th:prop_cvunif_g}, $(g_N)_{N\in\mathbb{N}^*}$ converges to $g$ uniformly on $\mathcal{U}$, so $(g_N(u_N))_{N\in\mathbb{N}^*}$ converges a.s. to $g(\Bar{u})$ and we deduce from~\eqref{eq:in_gamma_N}, that $g(\Bar{u}) \leq K$ so $(\Bar{u}) \in \mathcal{U}^K$.
\end{proof}

\begin{lemma}\label{th:lemma_no_empty_int}
    If the interior of $\mathcal{U}^K$ is non-empty, there exists $\Bar{u}  \in \mathcal{U}^K$ and $N' \in \mathbb{N}^*$ such that $g_N(\Bar{u}) < K$ for all $N>N'$.
\end{lemma}
\begin{proof}
    By assumption, $\mathcal{U}^K$ has a non-empty interior, so that there exists $\Bar{u}$ such that $g(\Bar{u}) < K$. We apply Proposition~\ref{th:prop_cvunif_g} to prove the a.s. uniform convergence of $(g_N)_{N\in\mathbb{N}^*}$ to $g$ on $\mathcal{U}$. Then, for $\varepsilon = K - g(\Bar{u}) > 0 $, there exists $N'$ such that $|g_N(\Bar{u}) - g(\Bar{u})| < \varepsilon $ for all $N>N'$. Therefore, $g_N(\Bar{u}) < K$.
\end{proof}

\begin{prp}\label{th:prop_cv_b}
    For any $u \in S $, there exists a sequence $(u_N)_{N \in \mathbb{N}^*}$ such that $(u_N) \in \mathcal{U}_N^K$ for all $N \in \mathbb{N}^*$ and $u_N \underset{N \rightarrow \infty}{\longrightarrow} u $ a.s.
\end{prp}
\begin{proof}
    Let us show this result for all $u \in \mathcal{U}^K$ which will give the desired result since $S \subset \mathcal{U}^K$. Let $u^* \in \mathcal{U}^K$. We construct a sequence $u_N \in \mathcal{U}_N^K$ which converges a.s. to $u^*$. Since the interior of $\mathcal{U}^K$ is non-empty, by Lemma~\ref{th:lemma_no_empty_int}, there exists $\Bar{u} =(\Bar{\boldsymbol{\gamma}},\Bar{\zeta}) \in \bigcap\limits_{N>N'} \mathcal{U}_N^K$ for some $N'>0$. For $\lambda \in [0,1]$, we define $\theta_N(\lambda) := g_N(\lambda \Bar{u} + (1-\lambda)u^*)$ and $\theta(\lambda) := g(\lambda \Bar{u} + (1-\lambda)u^*)$. Note that $\theta_N(1) = g_N(\Bar{u}) < K$ and $g(u^*) \leq K$.

    Let us define $\lambda_N = \operatorname*{\min} \{ \lambda \in [0,1]; \theta_N(\lambda) \leq K \}$ and $u_N = \lambda_N \Bar{u} + (1-\lambda_N)u^*$, $u_N \in \mathcal{U}_N^K$.

    \underline{\textbf{Case 1: $g(u^*) < K$}}: For any $\varepsilon > 0$, by Proposition~\ref{th:prop_cvunif_g} there exists $N'$ such that for all $N>N'$ we have $\sup\limits_{\mathcal{U}}|g_N - g| \leq \varepsilon$ so $\theta_N(0) = g_N(u^*) \leq K-\varepsilon$. We deduce that for $N$ large enough $\lambda_N = 0$ i.e. $u_N = u^*$ so that $u_N \rightarrow u^*$ a.s.

    \underline{\textbf{Case 2: $g(u^*) = K$}}: Note that the function $\theta_N$ is convex because $g_N$ is convex. So that $\theta_N(\lambda_N) \leq \theta_N(0) + \lambda_N(\theta_N(1) - \theta_N(0))$. As $\lambda_N$ is bounded, we may extract a convergent sub-sequence $\lambda_{\phi(N)} \rightarrow \Bar{\lambda}$ a.s. We assume that $\Bar{\lambda}>0$ and we use a proof by contradiction. We have:
    \begin{equation*}
        \theta_{\phi(N)}(\lambda_{\phi(N)}) \leq \theta_{\phi(N)}(0) + \lambda_N(\theta_{\phi(N)}(1) - \theta_{\phi(N)}(0))
    \end{equation*}
    taking the limit leads to
    \begin{equation*}
        \theta(\Bar{\lambda}) - \theta(0) \leq \Bar{\lambda} (\theta(1) - \theta(0)).
    \end{equation*}

    Now $\theta(0) = K$ and $\theta(1) < K$. Therefore $\theta(\Bar{\lambda})-K < 0$ which we rewrite as $\theta(\Bar{\lambda}) \leq K-3\varepsilon$ for some $\varepsilon >0$. By continuity of $\theta$, there exists $\delta > 0$ such that $\theta(\Bar{\lambda}-\delta) \leq K -2\varepsilon$. Note that $\sup_{\lambda} |\theta_N(\lambda)-\theta(\lambda)| \underset{N \rightarrow \infty}{\longrightarrow} 0$ a.s. because of  Proposition~\ref{th:prop_cvunif_g}. Then, for N large enough $|\theta_N(\lambda)-\theta(\lambda)|\leq \varepsilon$ for all $\lambda \in [0,1]$, so that $\theta_N(\Bar{\lambda}-\delta) \leq K-\varepsilon$. Thus, $\lambda_N \leq \Bar{\lambda}-\delta $ which contradicts that $\lambda_{\phi(N)} \rightarrow \Bar{\lambda}$ a.s. So that, $(\lambda_N)_{N\in \mathbb{N}^*}$ admits as unique limit $0$ and we conclude that $\lambda_N \rightarrow 0$ a.s., hence $u_N \rightarrow u^*$ a.s. which completes the proof.\newline
\end{proof}

\begin{thm}\label{th:thm_final_convergence}
    Let $v_N$ define by \eqref{eq:P} and $v^*$ define by \eqref{eq:P-eq-SAA}, if (P1) and (P2) hold, then
    \begin{equation}\label{eq:global_conv}
        \begin{aligned}
            v_N                &  &  & \overset{a.s.} {\underset{N \rightarrow \infty}{\longrightarrow}} &  & v^* \\
            \mathbb{D}(S_N, S) &  &  & \overset{a.s.} {\underset{N \rightarrow \infty}{\longrightarrow}} &  & 0.
        \end{aligned}
    \end{equation}
\end{thm}
\begin{proof}

    We want to apply Theorem~\ref{th:shapiro-5.5}. \\
    Assumption (a) holds thanks to Proposition~\ref{th:prop_cv_a} and Assumption (b) holds thanks to Proposition~\ref{th:prop_cv_b}.

    $\mathcal{U}^K$ has a non-empty interior, we can apply Lemma~\ref{th:lemma_no_empty_int} and find $\Bar{u} \in \mathcal{U}^K$ and $N' \in \mathbb{N}^*$ such that $g_N(\Bar{u}) < K$ for all $N>N'$. So that for $N>N', \ \Bar{u} \in \mathcal{U}_N^K$, then $\mathcal{U}_N^K$ is non-empty. Finally, $v_N$ is well defined and because $\mathcal{U}_N^K$ is closed, $S_N$ is included in $\mathcal{U}_N^K \subset \mathcal{U}$. Hence Assumption (c) holds.\\
    We can therefore apply Theorem~\ref{th:shapiro-5.5} and conclude the proof.
\end{proof}

\subsection{Uniqueness Condition}
Let us focus on the uniqueness of the solution to \eqref{eq:P} and therefore~\eqref{eq:P-eq}.

Let us define $\ell_0: \boldsymbol{\gamma} \mapsto \mathbb{E}(L_0(\boldsymbol{\gamma}, \textbf{X}))$. We shall make the following assumptions:
\begin{itemize}
    \item[(P3)] $\forall \boldsymbol{\gamma} \in \mathbb{R}^d_*$, $\boldsymbol{\gamma}^T \textbf{X}$ has a density on $\mathbb{R}^d$.
    \item[(P4)] $C_\alpha (\boldsymbol{\gamma})>K, \forall \boldsymbol{\gamma} \in \underset{\Gamma}{\arg\min} \ \ell_0$.
    \item[(P5)] For all $\boldsymbol{\gamma} \in \Gamma \setminus \{0_d\}$, $\text{supp}(\boldsymbol{\gamma}^T \textbf{X})$ is connected i.e. an interval.

\end{itemize}

\begin{remark}
    This setting is realistic, as capital naturally acts as a limiting constraint. Moreover, the support of $\boldsymbol{\gamma}^T \mathbf{X}$ is typically an interval, since aggregated losses across business lines generally have connected supports. Rare risk types with disconnected support, such as those related to natural disasters, are usually smoothed out by other continuous components. Assumption (P3) is satisfied when $\mathbf{X}$ admits a density, and (P4) ensures that the capital constraint is active at the optimum.
\end{remark}

Note that we exclude the case where $0_d$ is an optimal solution because we choose a positive capital requirement limit $K$.

Assumption (P2) implies $\mathbb{E}\left(\underset{\boldsymbol{\gamma} \in \Gamma}{\sup}|L_0(\boldsymbol{\gamma},\textbf{X})|\right)<+\infty$ and assumption (P1) implies \\
$\mathbb{P}( \boldsymbol{\gamma} \rightarrow L_0(\boldsymbol{\gamma}, \textbf{X}) \ \text{is continuous and convex} ) = 1$ since $L_0(\boldsymbol{\gamma},\textbf{X}) = L(\boldsymbol{\gamma},\Tilde{\zeta},\textbf{X})$ with $\Tilde{\zeta} \in A_\alpha(\gamma)$. Therefore, function $\ell_0$ is continuous and convex thanks to (P1) and (P2).

We are now in a position to state our main results. Theorem~\ref{th:thm_unicity} establishes the uniqueness of the solution to the original problem~\eqref{eq:P}, while Theorem~\ref{th:theorem_final_unicity} extends this result to the equivalent formulation~\eqref{eq:P-eq}.

Generally, the presence of a constraint does not guarantee uniqueness of the solution. However, in our setting, the strict convexity of the constraint function plays a key role in ensuring uniqueness. A crucial step in the analysis is to verify that the capital constraint is active at the optimum, allowing this strict convexity to be exploited. This is established through Lemma~\ref{th:lemma_reach_K} and Proposition~\ref{prop:strict_convexity_CVaR}.

Rather than proving Theorems~\ref{th:thm_unicity} and~\ref{th:theorem_final_unicity} directly, we first present and prove intermediate results from which they follow.

\begin{thm}\label{th:thm_unicity}
    If (P1) to (P4) hold, then Problem~\eqref{eq:P} has a unique solution $\boldsymbol{\gamma}^*$.
\end{thm}

\begin{thm}\label{th:theorem_final_unicity}
    If (P1) to (P5) hold, then Problem~\eqref{eq:P-eq} has a unique solution $(\boldsymbol{\gamma}^*,\zeta^*)$ and in addition $\boldsymbol{\gamma}^*$ is the unique solution to Problem~\eqref{eq:P}.
\end{thm}

We now present the intermediate results used in the proofs.

\begin{lemma}\label{th:lemma_reach_K}
    Let $\boldsymbol{\gamma}^*$ be an optimal solution to~\eqref{eq:P}. If assumptions (P1), (P2), and (P4) hold, then $C_\alpha(\boldsymbol{\gamma}^*) = K$.
\end{lemma}

\begin{proof}
    Suppose, for the sake of contradiction, that $C_\alpha(\boldsymbol{\gamma}^*) < K$. Let $\tilde{\boldsymbol{\gamma}} \in \underset{\Gamma}{\arg\min} \, \ell_0$.

    Define 
    \begin{equation*}
        \tilde{\lambda} = \max \left\{ \lambda \in [0,1] \,:\, C_\alpha(\lambda \tilde{\boldsymbol{\gamma}} + (1 - \lambda) \boldsymbol{\gamma}^*) \leq K \right\}.
    \end{equation*}
    By continuity of $C_\alpha$, we have $\tilde{\lambda} < 1$, since $C_\alpha(\tilde{\boldsymbol{\gamma}}) > K$ by (P4), and $C_\alpha(\boldsymbol{\gamma}^*) < K$ by assumption.

    Using the convexity of $\ell_0$ (from assumptions (P1) and (P2)), we obtain:
    \begin{equation*}
        \ell_0(\tilde{\lambda} \tilde{\boldsymbol{\gamma}} + (1 - \tilde{\lambda}) \boldsymbol{\gamma}^*) 
    \leq \tilde{\lambda} \ell_0(\tilde{\boldsymbol{\gamma}}) + (1 - \tilde{\lambda}) \ell_0(\boldsymbol{\gamma}^*) 
    < \ell_0(\boldsymbol{\gamma}^*),
    \end{equation*}
    where the strict inequality follows from the fact that $\boldsymbol{\gamma}^* \notin \underset{\Gamma}{\arg\min} \, \ell_0$ (by (P4)). This contradicts the optimality of $\boldsymbol{\gamma}^*$. Hence, we conclude that $C_\alpha(\boldsymbol{\gamma}^*) = K$.
\end{proof}

\begin{prp}\label{prop:strict_convexity_CVaR}
    Suppose that assumptions (P3) holds. Then the mapping $\boldsymbol{\gamma} \mapsto C_\alpha(\boldsymbol{\gamma})$ is strictly convex on $\mathbb{R}^d_*$.
\end{prp}

The proof of Proposition~\ref{prop:strict_convexity_CVaR} is provided in the appendix. 

We are now in a position to prove Theorem~\ref{th:thm_unicity}.

\begin{proof}[Proof of Theorem~\ref{th:thm_unicity}]
    Let $\boldsymbol{\gamma}_1, \boldsymbol{\gamma}_2$ be two optimal solutions of Problem~\eqref{eq:P}. By Lemma~\ref{th:lemma_reach_K}, both satisfy the capital constraint with equality:
    \begin{equation*}
        \ell_0(\boldsymbol{\gamma}_1) = \ell_0(\boldsymbol{\gamma}_2) = v^*, \quad 
        C_\alpha(\boldsymbol{\gamma}_1) = C_\alpha(\boldsymbol{\gamma}_2) = K.
    \end{equation*}
    Let $\boldsymbol{\gamma}_\mu := \mu \boldsymbol{\gamma}_1 + (1 - \mu) \boldsymbol{\gamma}_2$ for some $\mu \in (0,1)$. Since $C_\alpha$ is convex (by sub-additivity of CVaR), we have:
    \begin{equation*}
        C_\alpha(\boldsymbol{\gamma}_\mu) \leq \mu C_\alpha(\boldsymbol{\gamma}_1) + (1 - \mu) C_\alpha(\boldsymbol{\gamma}_2) = K.
    \end{equation*}
    Similarly, convexity of $\ell_0$ (by assumptions (P1) and (P2)) gives:
    \begin{equation*}
        \ell_0(\boldsymbol{\gamma}_\mu) \leq \mu \ell_0(\boldsymbol{\gamma}_1) + (1 - \mu) \ell_0(\boldsymbol{\gamma}_2) = v^*.
    \end{equation*}
    If $\ell_0(\boldsymbol{\gamma}_\mu) < v^*$, this contradicts the optimality of $\boldsymbol{\gamma}_1$ and $\boldsymbol{\gamma}_2$. Thus, $\ell_0(\boldsymbol{\gamma}_\mu) = v^*$, and $\boldsymbol{\gamma}_\mu$ is optimal. Applying Lemma~\ref{th:lemma_reach_K} again yields $C_\alpha(\boldsymbol{\gamma}_\mu) = K$.

    Therefore, for all $\mu \in [0,1]$, we have:
    \begin{equation*}
        C_\alpha(\boldsymbol{\gamma}_\mu) = \mu C_\alpha(\boldsymbol{\gamma}_1) + (1 - \mu) C_\alpha(\boldsymbol{\gamma}_2),
    \end{equation*}
    which implies equality in the convexity inequality. By the strict convexity of $C_\alpha$ (Proposition~\ref{prop:strict_convexity_CVaR}), it follows that $\boldsymbol{\gamma}_1 = \boldsymbol{\gamma}_2$. Hence, the solution to Problem~\eqref{eq:P} is unique.
\end{proof}

\begin{prp}\label{th:prop_convexity}
    For any fixed $\boldsymbol{\gamma} \neq 0_d$, if (P3) and (P5) holds, the function $\zeta \mapsto g(\boldsymbol{\gamma}, \zeta)$ is strictly convex on $\text{supp}(-\boldsymbol{\gamma}^T \textbf{X})$.
\end{prp}
\begin{proof}
    Consider the derivative $\frac{\partial g}{\partial \zeta}(\boldsymbol{\gamma}, \zeta) = 1 - (1-\alpha)^{-1}\mathbb{E}(-\vmathbb{1}_{\{ -\boldsymbol{\gamma}^T \textbf{X} - \zeta > 0 \}}) = 1 - (1-\alpha)^{-1}\mathbb{P}( -\boldsymbol{\gamma}^T \textbf{X} > \zeta )$.\newline
    Let us $\zeta_1 < \zeta_2 \in \text{supp}(-\boldsymbol{\gamma}^T \textbf{X})$
    \begin{align*}
        \frac{\partial g}{\partial \zeta}(\boldsymbol{\gamma}, \zeta_2) - \frac{\partial g}{\partial \zeta}(\boldsymbol{\gamma}, \zeta_1) & = (1-\alpha)^{-1}\left(\mathbb{P}( -\boldsymbol{\gamma}^T \textbf{X} > \zeta_1 ) - \mathbb{P}( -\boldsymbol{\gamma}^T \textbf{X} > \zeta_2 )\right) \\
                                                                                                                                          & =(1-\alpha)^{-1}\mathbb{P}( \zeta_1 < -\boldsymbol{\gamma}^T \textbf{X} < \zeta_2) \geq 0
    \end{align*}
    this quantity is non negative on $\mathbb{R}$ and positive on $\text{supp}(-\boldsymbol{\gamma}^T \textbf{X})$ because $-\boldsymbol{\gamma}^T \textbf{X}$ has a density by (P3) and $\text{supp}(-\boldsymbol{\gamma}^T \textbf{X})$ is an interval thanks to (P5).
\end{proof}

Finally, we are going to prove Theorem~\ref{th:theorem_final_unicity}.
\begin{proof}[Proof of Theorem~\ref{th:theorem_final_unicity}]
    Using Theorem~\ref{th:thm_unicity}, let $\boldsymbol{\gamma}^*$ be the unique solution to Problem~\eqref{eq:P}. The equivalence between the Problems~\eqref{eq:Reward_True} and~\eqref{eq:Reward_True_F} implies that any $(\tilde{\boldsymbol{\gamma}}, \tilde{\zeta})$ solution to \eqref{eq:P-eq} satisfies $\tilde{\boldsymbol{\gamma}} = \boldsymbol{\gamma}^*$.

    Moreover $C_\alpha(\boldsymbol{\gamma}^*) = \underset{\boldsymbol{\zeta} \in \mathbb{R}}{\min} \ g(\boldsymbol{\gamma}^*, \zeta)$ and thanks to Proposition~\ref{th:prop_convexity}, $A_\alpha(\gamma^*)$ = $\{\zeta^*\}$. Therefore $(\boldsymbol{\gamma}^*,\zeta^*)$ is the unique solution of Problem~\eqref{eq:P-eq}.
\end{proof}

\subsection{Convergence rate}
In Shapiro, A., Dentcheva, D., Ruszczynski, A. (2009) \cite{shapiro2009:8} the rate of convergence is also obtained. We state a version (Theorem \ref{th:shapiro-5.11}), adapted to our setting. It is used to get the convergence rate in Theorem \ref{th:global_rate_cv}.

Consider some additional notations and assumption:
\begin{itemize}
    \item[(P6)] $\mathbb{E}(|\textbf{X}|^2) < + \infty$.
    \item[(P7)] $\mathbb{E}\left(\underset{(\boldsymbol{\gamma},\zeta) \in \mathcal{U}}{\sup}|L(\boldsymbol{\gamma},\zeta,\textbf{X})|^2\right)<+\infty$.
    \item[(P8)] There exists a measurable function $C_L : X \rightarrow \mathbf{R^+}$ such that $\mathbb{E}(C_L(\textbf{X})^2) < \infty$ and  $\forall u, u' \in \mathcal{U}, |L(u,\textbf{X})-L(u',\textbf{X})| \leq C_L(\textbf{X})|u-u'|$ almost surely.
    \item[(P9)] $(\boldsymbol{\gamma},\zeta) \rightarrow \ell(\boldsymbol{\gamma},\zeta)$ is differentiable.
\end{itemize}

\begin{remark}
    {
        For loss functions of the form $L(\boldsymbol{\gamma}, \zeta, \mathbf{X}) = \mathbb{E}\left(-\boldsymbol{\gamma}^T \mathbf{X} + \text{c} \cdot \left(\zeta + \frac{1}{1 - \alpha} (-\boldsymbol{\gamma}^T \mathbf{X} - \zeta)^+ \right) \right)$, or $L(\boldsymbol{\gamma}, \mathbf{X}) = \mathbb{E}\left(-U(\boldsymbol{\gamma}^T \mathbf{X}) \right)$ where $U$ is a convex utility function, it is straightforward to establish the existence of the associated function $C_L$.} 
\end{remark}

Theorem [5.11 p173 \cite{shapiro2009:8}] deals with problems of the following form:
\begin{equation}\label{eq:th5.11}
    \begin{aligned}
         & \operatorname*{\min}_{u \in \mathcal{U}} \ell_N(u) \ \text{ s.t. } \  g_{iN}(u) \leq 0, \ i = 1,..,p.
    \end{aligned}
\end{equation}

For Problem~\eqref{eq:P-eq} there are two kinds of constraints, one on the parameters $\boldsymbol{\gamma}, \zeta$ and one with the $g$ function. Consider $g_1: (\boldsymbol{\gamma},\zeta) \mapsto g(\boldsymbol{\gamma},\zeta) - K, g_2:\boldsymbol{\gamma} \mapsto \boldsymbol{\gamma}-\boldsymbol{\gamma}^{up}$ and $g_3:\boldsymbol{\gamma} \mapsto \boldsymbol{\gamma}^{low}-\boldsymbol{\gamma}$. We have $G_1: ((\boldsymbol{\gamma},\zeta),\textbf{X}) \mapsto G((\boldsymbol{\gamma},\zeta),\textbf{X}) - K = \zeta + (1-\alpha)^{-1}(-\boldsymbol{\gamma}^T \textbf{X} - \zeta)^+ - K, G_2 = g_2$ and $G_3 = g_3$. They are defined such that for any $u \in \mathcal{U} $ we have $g_i(u) = \mathbb{E}(G_i(u, \textbf{X}))$. We denote by $g_{iN}$ the approximated version of $g_i$.

Let $u\in \mathcal{U}$, then the empirical mean estimator $\ell_N(u)$ of $\ell(u)$ is unbiased and has a variance of $\sigma^2(u) := \mathbb{V}ar(L(u,\textbf{X}))$ which is assumed to be finite. Moreover the standard Central Limit Theorem gives:
\begin{equation*}
    N^{1/2} \left( \ell_N(u) - \ell(u) \right) \underset{N \rightarrow +\infty}{\overset{\mathcal{D}} {\rightarrow}} \mathcal{N}(0,\sigma^2(u))
\end{equation*}

where $\overset{\mathcal{D}} {\rightarrow}$ denotes the convergence in distribution. Let $Y(u)$ be a random variable such that $Y(u) \sim \mathcal{N}(0,\mathbb{V}ar(L(u,\textbf{X})))$. Let $Y_i(u)$ be a random variable such that $Y_i(u) \sim \mathcal{N}(0,\mathbb{V}ar(G_i(u,\textbf{X}))), \ i=1,2,3$.

Let us introduce the Lagrangian of Problem~\eqref{eq:th5.11}:
\begin{equation*}
    \mathcal{L}(u,\lambda) := \ell(u) + \sum_{i=1}^p \lambda_i g_i(u),
\end{equation*}
which can be rewritten as
\begin{equation}\label{eq:lagrange_def}
    \mathcal{L}(\boldsymbol{\gamma}, \zeta, (\lambda,\Bar{\boldsymbol{\mu}}, \underline{\boldsymbol{\mu}})) := \ell(\boldsymbol{\gamma}, \zeta) + \lambda (g(\boldsymbol{\gamma}, \zeta)-K) + \Bar{\boldsymbol{\mu}}.(\boldsymbol{\gamma}-\boldsymbol{\gamma}^{up}) - \underline{\boldsymbol{\mu}}.(\boldsymbol{\gamma}-\boldsymbol{\gamma}^{low}).
\end{equation}

In our context, Problem~\eqref{eq:P-eq} is convex because $\ell$ is convex. In addition the set $S$ of optimal solutions to the original problem~\eqref{eq:P} is non-empty and bounded. Moreover $\ell$ and $g_i$ are bounded on $\mathcal{U}$ and in effect in a neighbourhood of $S$. Thanks to these results, we can state a modified version of Theorem 5.11 p173 in \cite{shapiro2009:8} adapted to our framework.
\begin{thm}[Restatement of Theorem 5.11 p173 in \cite{shapiro2009:8}]\label{th:shapiro-5.11}\par
    Let $v_N$ be the optimal value of the above problem~\eqref{eq:th5.11}. Assume that the following assumptions are satisfied:

    \begin{enumerate}
        \item[(i)] There exists $\Bar{u} \in \mathcal{U}$ such that $g_i(\Bar{u}) < 0, \ i= 1,2,3$. (Slater's condition).
        \item[(ii)] The following two assumptions are satisfied for $H = L, G_i, i = 1,2,3$:
              \begin{enumerate}
                  \item[($A_1$)] There exists $\Tilde{u} \in \mathcal{U}$ such that $\mathbb{E}(H(u,\textbf{X})^2)<+\infty$.
                  \item[($A_2$)] There exists a measurable function $C : X \rightarrow \mathbf{R^+}$ such that $\mathbb{E}(C(\textbf{X})^2) < \infty$ and  $\forall u, u' \in \mathcal{U}, |H(u,\textbf{X})-H(u',\textbf{X})| \leq C(\textbf{X})|u-u'|$ almost surely.
              \end{enumerate}
    \end{enumerate}
    Then
    \begin{equation}\label{eq:result_rate_cv}
        N^{1/2} \left( v_N - v^* \right) \overset{\mathcal{D}} {\rightarrow} \inf_{u \in S} \sup_{\boldsymbol{\lambda \in \Lambda}} \left[ Y(u) + \sum_{i=1}^p \lambda_i Y_i(u) \right].
    \end{equation}

    A set $\Lambda$ of Lagrange vectors $\lambda = (\lambda_1, .., \lambda_p)$ is associated to each optimal solution $u^* \in S$ which satisfy,

    \begin{equation*}
        u^* \in \operatorname*{arg} \operatorname*{\min}_{\boldsymbol{u \in \mathcal{U}}} \mathcal{L}(u,\lambda), \ \lambda_i \leq 0 \ \text{and} \ \lambda_i g_i(u^*) = 0, \ i=1,..,p.
    \end{equation*}

    The set $\Lambda$ coincides with the set of optimal solutions to the dual of the original problem and is therefore the same for any optimal solution $u^* \in S$.

    Furthermore, if $S = \{ u^* \}$ and $\Lambda =\{\lambda^*\}$ are singletons, then
    \begin{equation*}
        N^{1/2} \left( v_N - v^* \right) \overset{\mathcal{D}} {\rightarrow} \mathcal{N}(0,\sigma^2)
    \end{equation*}
    with
    \begin{equation*}
        \sigma^2 := \mathbb{V}ar \left( L(u^*,\textbf{X}) + \sum_{i=1}^p \lambda^*_i G_i(u^*,\textbf{X}) \right).
    \end{equation*}
\end{thm}

\begin{lemma}\label{th:lemma_extremal_point}
    Given a set of constraints $\boldsymbol{\gamma}^{up}, \boldsymbol{\gamma}^{low}$ and $\boldsymbol{\gamma} \in \Gamma$, if (P1), (P2) hold, then there exists at most a single $\displaystyle K \leq \min_{\gamma \in \underset{\Gamma}{\arg\min} \ \ell_0} C_\alpha(\gamma)$ such that $\boldsymbol{\gamma}$ solves \eqref{eq:P} with $K$ as the capital limit.
\end{lemma}
\begin{proof}
    This is a corollary of Lemma \ref{th:lemma_reach_K}, because if $\boldsymbol{\gamma}$ is solution of \eqref{eq:P} with $K' \leq C_\alpha(\underset{\Gamma}{\arg\min} \ \ell_0)$, then $K' = C_\alpha(\boldsymbol{\gamma})$ and so there can only be one $K'$.
\end{proof}

\begin{thm}\label{th:global_rate_cv}
    If (P1) to (P9) hold, denote by $u^* = (\boldsymbol{\gamma}^*, \zeta^*)$ the unique solution to Problem~\eqref{eq:P-eq}, then
    \begin{equation*}
        N^{1/2} \left( v_N - v^* \right) \overset{\mathcal{D}} {\rightarrow} \mathcal{N}(0,\sigma^2)
    \end{equation*}
    with
    \begin{equation*}
        \sigma^2 := \mathbb{V}ar \left( L(u^*,\textbf{X}) + \lambda^*( G(u^*,\textbf{X})-K) \right)
    \end{equation*}
    and $\lambda^*$ is the maximum value of the Lagrange multiplier $\lambda$ associated with the capital constraint.
\end{thm}

\begin{proof}[Proof of Theorem \ref{th:global_rate_cv}]
    The uniqueness of $u^*$ is given by (P1) to (P5) using Theorem \ref{th:theorem_final_unicity}.
    Let us look at the conditions for applying Theorem~\ref{th:shapiro-5.11}. Condition (i) holds because we assumed that the interior of $\mathcal{U}^K$ is non-empty which means that Slater's condition is satisfied for $g_1,g_2,g_3$. The functions $G_2,G_3$ are deterministic and linear functions of $\boldsymbol{\gamma}$ so ($A_1$) and ($A_2$) hold. Conditions ($A_1$) and ($A_2$) hold for $L$ by (P7) and (P8) and for $G_1$ it is easy to prove with (P6).

    We can apply Theorem~\ref{th:shapiro-5.11} to obtain~\eqref{eq:result_rate_cv}. The functions $G_2$ and $G_3$ are deterministics so $Y_2 = Y_3 = 0$, therefore $ \sum_{i=1}^p \lambda_i Y_i((\boldsymbol{\gamma}, \zeta)) = \lambda_1 Y_1((\boldsymbol{\gamma}, \zeta))$. Thanks to the uniqueness of $u^*$, we only need the maximum value of $\lambda_1$ or the uniqueness of $\lambda_1$ to apply the second part of Theorem~\ref{th:shapiro-5.11}.

    Consider the Lagrange multiplier~\eqref{eq:lagrange_def} and recall that $g(\boldsymbol{\gamma}, \zeta) = \zeta + (1-\alpha)^{-1} \mathbb{E}\left[ (-\boldsymbol{\gamma}^T \textbf{X} - \zeta)^+ \right]$.

    The set $\Lambda$ of Lagrange multipliers vectors $(\lambda, \Bar{\boldsymbol{\mu}}, \underline{\boldsymbol{\mu}})$ is satisfying the optimality conditions. The optimal solution $(\boldsymbol{\gamma}^*, \zeta^*)$ satisfies $\nabla_{\boldsymbol{\gamma},\zeta} \mathcal{L}(\boldsymbol{\gamma}^*, \zeta^*) = 0$ and the following optimality conditions $\lambda (g(\boldsymbol{\gamma}^*, \zeta^*)-K) = 0$, $\Bar{\boldsymbol{\mu}}(\boldsymbol{\gamma}^*-\boldsymbol{\gamma}^{up}) = 0$ and $\underline{\boldsymbol{\mu}}(\boldsymbol{\gamma}^*-\boldsymbol{\gamma}^{low}) = 0$. Note that $\mathbb{P}(-\boldsymbol{\gamma}^{* T}\textbf{X}-\zeta^* \geq 0 ) = 1-\alpha$ because $\zeta^* \in A_\alpha(\gamma^*)$.

    Consider
    \begin{align*}
        \frac{\partial \mathcal{L}}{\partial \zeta}(\boldsymbol{\gamma}^*, \zeta^*)
         & = \frac{\partial \ell}{\partial \zeta}(\boldsymbol{\gamma}^*, \zeta^*) + \lambda \left( 1 + (1-\alpha)^{-1} \mathbb{E}(-\vmathbb{1}_{\{ -\boldsymbol{\gamma}^{* T}\textbf{X}-\zeta^* \geq 0 \}}) \right) \\
         & = \frac{\partial \ell}{\partial \zeta}(\boldsymbol{\gamma}^*, \zeta^*) + \lambda \left( 1 - (1-\alpha)^{-1} \mathbb{P}( -\boldsymbol{\gamma}^{* T}\textbf{X}-\zeta^* \geq 0 ) \right)                    \\
         & =\frac{\partial \ell}{\partial \zeta}(\boldsymbol{\gamma}^*, \zeta^*).
    \end{align*}

    Thus,
    \begin{equation}\label{eq:lagrange_L_by_zeta_zero}
        \frac{\partial \ell}{\partial \zeta}(\boldsymbol{\gamma}^*, \zeta^*) = 0.
    \end{equation}
    Consider
    \begin{align*}
        \frac{\partial \mathcal{L}}{\partial \boldsymbol{\gamma}}(\boldsymbol{\gamma}^*, \zeta^*)
         & = \frac{\partial \ell}{\partial \boldsymbol{\gamma}}(\boldsymbol{\gamma}^*, \zeta^*) + \lambda(1-\alpha)^{-1} \mathbb{E}(-\textbf{X}. \vmathbb{1}_{\{ -\boldsymbol{\gamma}^{* T}\textbf{X}-\zeta^* \geq 0 \}}) + \Bar{\boldsymbol{\mu}} - \underline{\boldsymbol{\mu}}                                                \\
         & = \frac{\partial \ell}{\partial \boldsymbol{\gamma}}(\boldsymbol{\gamma}^*, \zeta^*) + \lambda(1-\alpha)^{-1} \mathbb{E}(-\textbf{X} | -\boldsymbol{\gamma}^{* T}\textbf{X}-\zeta^* \geq 0 )\mathbb{P}(-\boldsymbol{\gamma}^{* T}\textbf{X}-\zeta^* \geq 0 ) + \Bar{\boldsymbol{\mu}} - \underline{\boldsymbol{\mu}}.
    \end{align*}

    So for each $i=1,..,d$:
    \begin{equation}\label{eq:Lagrangian_partial_gamma}
        \frac{\partial \ell}{\partial \boldsymbol{\gamma}_i}(\boldsymbol{\gamma}^*, \zeta^*) + \lambda \mathbb{E}(- \textbf{X}_i | -\boldsymbol{\gamma}^{* T}\textbf{X}-\zeta^* \geq 0) + \Bar{\mu_i} - \underline{\mu_i} = 0.
    \end{equation}

    We want to find $j \in \{1, \ldots, d \}$ such that  $\frac{\partial \ell}{\partial \boldsymbol{\gamma}_j}(\boldsymbol{\gamma}^*, \zeta^*) \neq 0$ and $\Bar{\mu_j} = \underline{\mu_j} = 0$.
    Let $H \subset \{1,\ldots ,d\}$ such that $\forall j \in H$, $\frac{\partial \ell}{\partial \boldsymbol{\gamma}_j}(\boldsymbol{\gamma}^*, \zeta^*) \neq 0$.

    Let $\tilde{\boldsymbol{\gamma}} \in \underset{\Gamma}{\arg\min} \ \ell_0$ and $\tilde{\zeta} \in V_\alpha(\tilde{\boldsymbol{\gamma}})$, so $g(\tilde{u}) = C_\alpha (\tilde{\boldsymbol{\gamma}})> K$ by (P4), with $\tilde{u} = (\tilde{\boldsymbol{\gamma}}, \tilde{\zeta})$. Hence $\nabla \ell(u^*) \cdot (\Tilde{u}-u^*) <0$  thanks to the convexity of $\ell$ by (P1) and (P2). We use \eqref{eq:lagrange_L_by_zeta_zero} to obtain
    \begin{equation*}
        \nabla \ell(u^*).(\Tilde{u}-u^*) = \nabla \ell(\boldsymbol{\gamma}^*, \zeta^*) \cdot (\tilde{\boldsymbol{\gamma}}-\boldsymbol{\gamma}^*)< 0.
    \end{equation*}
    Thus, there exists $j \in \{1,\ldots, d\}$ such that $(\tilde{\boldsymbol{\gamma}}_j - \boldsymbol{\gamma}^*_j) \frac{\partial \ell}{\partial \boldsymbol{\gamma}_j}(\boldsymbol{\gamma}^*, \zeta^*) < 0$ and so $H \neq \emptyset$.

    There are three possible cases, with the optimality conditions, for $j\in H$:\\
    1. $\boldsymbol{\gamma}_j = \gamma_{j}^{up}$ implies $\underline{\boldsymbol{\mu}_j} = 0$ and  $\Bar{\boldsymbol{\mu}_j} \geq 0$;\\
    2. $\boldsymbol{\gamma}_j = \gamma_{j}^{low}$ implies $\underline{\boldsymbol{\mu}_j} \geq 0$ and  $\Bar{\boldsymbol{\mu}_j} = 0$;\\
    3. $ 0 < \gamma_{j}^{low} < \boldsymbol{\gamma}_j < \gamma_{j}^{up}$ implies $\underline{\boldsymbol{\mu}_j} = \Bar{\boldsymbol{\mu}_j} = 0$.

    If for all $j \in H$, $\boldsymbol{\gamma}_j = \gamma_{j}^{up}$ or $\boldsymbol{\gamma}_j = \gamma_{j}^{low}$ then we cannot conclude but we decided to exclude this case because it is unlikely in a sense that for a sufficiently small $\epsilon > 0$, for all $K'$ such that $|K-K'|<\epsilon$, the corresponding optimal solution $\boldsymbol{\gamma}^*_{K'}$ is not in this case. Indeed, thanks to \ref{th:lemma_extremal_point} the optimal solution cannot be the same as for $K$, and for all $j \notin H$, $\frac{\partial \ell}{\partial \boldsymbol{\gamma}_j}(\boldsymbol{\gamma}^*, \zeta^*) = 0$. Therefore, it is a $\gamma_i$ with an $i \in H$, so in case 1 or 2, which must move, which places it directly in case 3.

    So there exists $j\in H$ in the case 3, knowing that $\frac{\partial \ell}{\partial \boldsymbol{\gamma}_j}(\boldsymbol{\gamma}^*, \zeta^*) \neq 0$, $\Bar{\mu_j} = \underline{\mu_j} = 0$ and thanks to~\eqref{eq:Lagrangian_partial_gamma}, $\mathbb{E}(- \textbf{X}_j | -\boldsymbol{\gamma}^{* T}\textbf{X}-\zeta^* \geq 0) \neq 0$. Thus $\lambda = -\frac{\partial \ell}{\partial \boldsymbol{\gamma}_j}(\boldsymbol{\gamma}^*, \zeta^*) \mathbb{E}(- \textbf{X}_j | -\boldsymbol{\gamma}^{* T}\textbf{X}-\zeta^* \geq 0)^{-1}$. Therefore the Lagrange multiplier $\lambda$ is unique.

    Moreover, denoting this Lagrange multiplier by  $\lambda^*$, $\mathbb{V}ar \left( L(u^*,\textbf{X}) + \sum_{i=1}^p \lambda^*_i G_i(u^*,\textbf{X}) \right)$ becomes  $\mathbb{V}ar \left( L(u^*,\textbf{X}) + \lambda^*( G(u^*,\textbf{X})-K) \right)$.
\end{proof}

\section{Numerical study}\label{section:numerical_study}

This section illustrates the results of section \ref{section:main_results}, in particular convergence and convergence rate. When it was possible to compare the results of the SAA method with an explicit and computational result. We first consider the Gaussian case for which an explicit solution is known for the mean-CVaR, where all the asset returns follow a Gaussian distribution. This first case will allow us to validate the algorithm by comparing it with theoretical solutions. Second, we consider a more realistic case where all the asset returns do not follow the same distribution. This last case is closed to what can be found in (re)insurance.

All the examples are solved within a few seconds and are carried out using Python libraries (mainly Scipy, Numpy). This calculation time can be compared to solving using linear programming, see \cite{Uryasev2002:9}, which takes a long time as soon as the data sample is large because each element in the sample adds a dimension to the resolution.

\subsection{Context}

In what follows, we consider $L_0(\boldsymbol{\gamma},\textbf{X}) = -\boldsymbol{\gamma}^T \textbf{X}$ with a risk level $\alpha = 0.99$.

Note that $L_0$ is convex and continuous so (P1) hold.

Let us rewrite $L(\boldsymbol{\gamma},\zeta, \textbf{X}) = -\boldsymbol{\gamma}^T \textbf{X}$ and $\ell_N(\boldsymbol{\gamma},\zeta) := -\frac{\boldsymbol{\gamma}^T}{N} \sum_{j=1}^N \textbf{X}^{(j)}$.

Once we have chosen the asset distributions and other parameter for the problem  ($K, \gamma^{up}, \gamma^{low}$), for a given $N$, we run k optimizations, i.e. solve the SAA problem \eqref{eq:P-eq-SAA}, on k different $\textbf{X}$ i.i.d. sample of size N to obtain a sample of size k of $\boldsymbol{\gamma}_{N}, v_{N}$.

\subsection{Gaussian case}\label{subsection:Gaussian}
Let us consider the Gaussian case, $\textbf{X} \sim \mathcal{N}_d(\textbf{1}_d,\Sigma)$, with $\Sigma$ a positive definite matrix. Note that $\boldsymbol{\gamma}^T \textbf{X} \sim \mathcal{N}(\boldsymbol{\gamma}^T \textbf{1}_d, \boldsymbol{\gamma}^T \Sigma \boldsymbol{\gamma})$ and $CVaR_\alpha(-\boldsymbol{\gamma}^T \textbf{X}) = - \boldsymbol{\gamma}^T \textbf{1}_d + T_Z\sqrt{\boldsymbol{\gamma}^T \Sigma \boldsymbol{\gamma}} $, with $T_Z = CVaR_\alpha(\mathcal{N}(0,1))$.

If  $\tilde{\textbf{X}} \sim \mathcal{N}(\tilde{\mu},\tilde{\Sigma})$, then $\boldsymbol{\gamma}^T \tilde{\textbf{X}} = (\boldsymbol{\gamma}*\tilde{\mu})^T \textbf{X}$ and with $\textbf{X} \sim \mathcal{N}(\textbf{1}_d,\Sigma)$ and $\Sigma = \tilde{\mu}^T \tilde{\Sigma}  \tilde{\mu}$. Hence considering $\mathbb{E}(X) = \textbf{1}_d$ is sufficiently general.

Denote $\boldsymbol{\gamma}^*$ the optimal solution to~\eqref{eq:P}, if $\gamma^{up}_i < \gamma_i^* < \gamma^{low}_i $ for $i=1,\dots,d$, then:
\begin{equation}\label{eq:Gaussian_true_value}
    \boldsymbol{\gamma}^* = K \left( \frac{1}{T_Z-\sqrt{\textbf{1}_d^T \Sigma^{-1} \textbf{1}_d}} \right) \frac{\Sigma^{-1} \textbf{1}_d }{\sqrt{\textbf{1}_d^T \Sigma^{-1} \textbf{1}_d}},
\end{equation}
See Appendix 2 for a proof.

Properties (P2), (P5) and (P6) hold immediately. We choose the parameters bounds for (P4) to hold: $C_\alpha (\boldsymbol{\gamma})>K, \forall \boldsymbol{\gamma} \in \underset{\boldsymbol{\gamma} \in \Gamma}{\arg\min} \ \ell_0(\boldsymbol{\gamma})$. If not the capital limit is not reached. Note that (P7) to (P9) hold too.

We set $d=5$, $\alpha = 0.99$, $K=100$, $k=5000$ runs and covariance matrix $\Sigma$ picked at random.

\begin{equation*}
    \Sigma =
    \begin{pmatrix}
        4.490  & -0.377 & 0.059  & 0.585  & -1.709 \\
        -0.377 & 6.109  & -1.300 & 0.229  & 1.380  \\
        0.059  & -1.300 & 7.059  & -1.401 & 0.210  \\
        0.585  & 0.229  & -1.401 & 8.400  & -1.250 \\
        -1.709 & 1.380  & 0.210  & -1.250 & 19.934 \\
    \end{pmatrix}.
\end{equation*}

\subsubsection{Without bounds}\label{subsection:gauss_without_bound}
This case is without bounds in the sense that we set $\gamma^{up}, \gamma^{low}$ to have $\gamma^{up}_i < \gamma_i^* < \gamma^{low}_i $, in other words in such a way that they are never reached. So, we use \eqref{eq:Gaussian_true_value} to compute $\gamma^* = (15.113, 12.576, 12.610,  8.702,  3.958)$ and $v^* = -52.960$.

We also compute the standard estimator of $\sigma^2 = \mathbb{V}ar \left( - \boldsymbol{\gamma}^{*T} \textbf{X} + \lambda^*( \zeta^* + \frac{1}{1-\alpha}(- \boldsymbol{\gamma}^{*T}\textbf{X} -\zeta^*)^+ -K) \right) $ from Theorem \ref{th:global_rate_cv} and we obtain $\hat{\sigma}= 149.860$.

Figures \ref{fig:gamma_1} to \ref{fig:gamma_5} and \ref{fig:boxplot_vn} illustrate Theorem \ref{th:thm_final_convergence}, showing that the solution to the SAA formulation \eqref{eq:P-eq-SAA} converges to the original solution found with \eqref{eq:Gaussian_true_value}. Indeed, the convergence of $\boldsymbol{\gamma}_N$ and $v_N$ is observed and a sample size of $10^5$ is sufficient to obtain an acceptable error, here 2\% on each components of $\boldsymbol{\gamma}^*$ and 1\% on $v^*$.

Figure \ref{fig:std_vn} illustrates a convergence rate of $\sqrt{N}$ seen in Theorem \ref{th:global_rate_cv} since the slope of each curve is $-0.5$. We see empirically the same results for the parameters with figure \ref{fig:std_gamma_1} to \ref{fig:std_gamma_5}.

The histogram of Figure \ref{fig:gauss_d_3} has a Gaussian shape with parameters $(\mu,\sigma) = -9.15, 160.33$. This illustrates Theorem \ref{th:global_rate_cv} with $\hat{\sigma} = 149.860$. From Figure \ref{fig:gauss_d_2_bis}, we may infer the the parameters have an asymptotic gaussian behavior. 

\begin{figure}[H]
    \centering
    \begin{subfigure}{0.3\textwidth}
        \includegraphics[width=\textwidth]{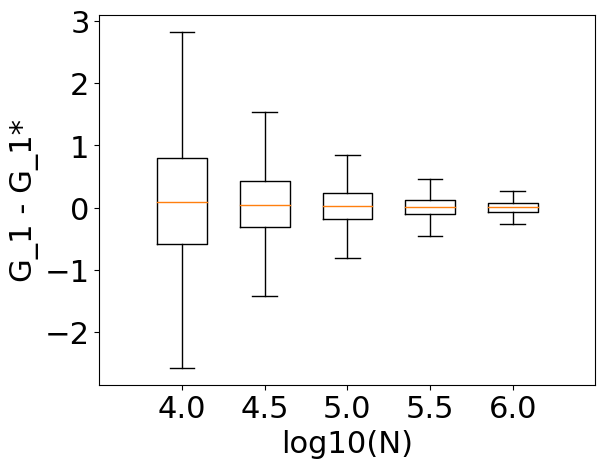}
        \caption{$\boldsymbol{\gamma}_{N,1}-\boldsymbol{\gamma}^{*}_1$}
        \label{fig:gamma_1}
    \end{subfigure}
    \hfill
    \begin{subfigure}{0.3\textwidth}
        \includegraphics[width=\textwidth]{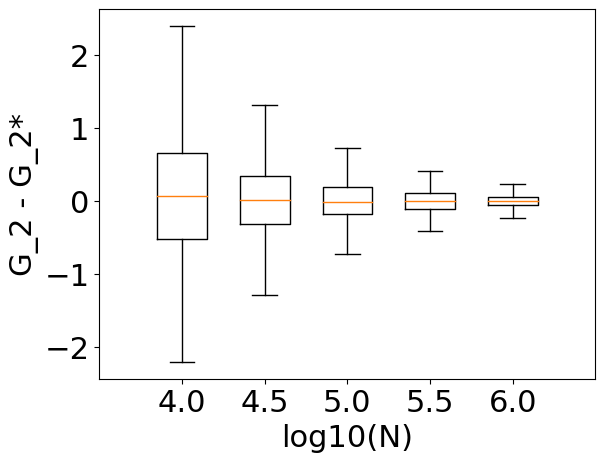}
        \caption{$\boldsymbol{\gamma}_{N,2}-\boldsymbol{\gamma}^{*}_2$}
        \label{fig:gamma_2}
    \end{subfigure}
    \hfill
    \begin{subfigure}{0.3\textwidth}
        \includegraphics[width=\textwidth]{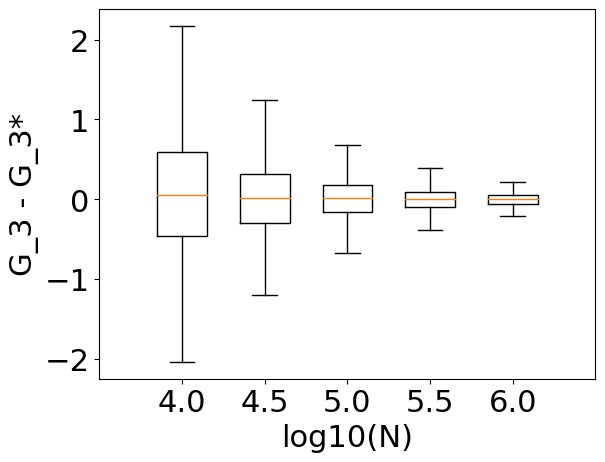}
        \caption{$\boldsymbol{\gamma}_{N,3}-\boldsymbol{\gamma}^{*}_3$}
        \label{fig:gamma_3}
    \end{subfigure}
    \begin{subfigure}{0.3\textwidth}
        \includegraphics[width=\textwidth]{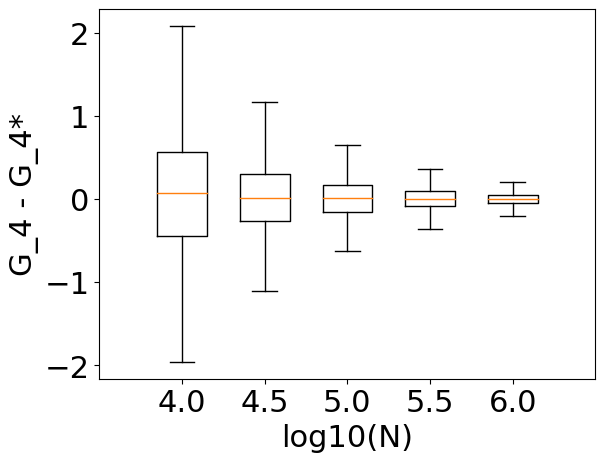}
        \caption{$\boldsymbol{\gamma}_{N,4}-\boldsymbol{\gamma}^{*}_4$}
        \label{fig:gamma_4}
    \end{subfigure}
    \begin{subfigure}{0.3\textwidth}
        \includegraphics[width=\textwidth]{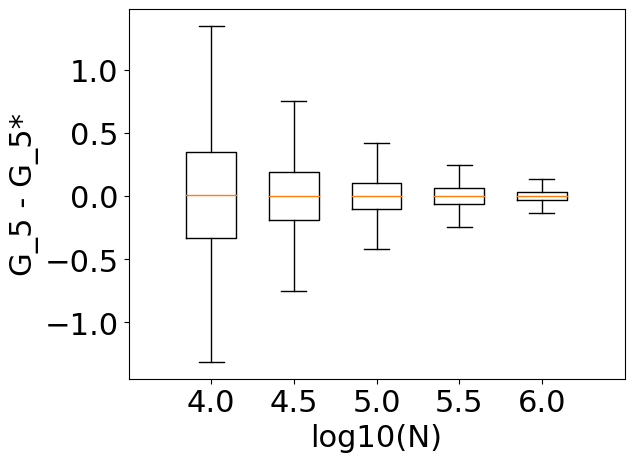}
        \caption{$\boldsymbol{\gamma}_{N,5}-\boldsymbol{\gamma}^{*}_5$}
        \label{fig:gamma_5}
    \end{subfigure}

    \caption{Component-wise Boxplots of $\boldsymbol{\gamma}_{N}-\boldsymbol{\gamma}^{*}$. (boxplot standard parameters: Q1-1.5IQR,Q1,median,Q3,Q3+1.5IQR)}
    \label{fig:gauss_d_1}
\end{figure}

\begin{figure}[H]
    \centering
    \begin{subfigure}{0.32\textwidth}
        \includegraphics[width=\textwidth]{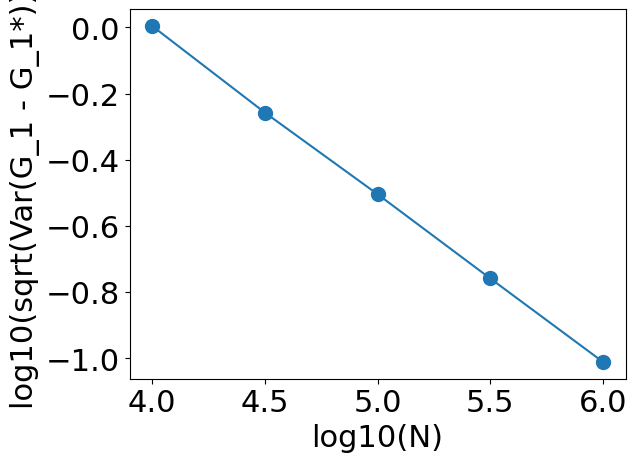}
        \caption{$\log_{10} \left( \sqrt{Var(\gamma_{N,1}-\gamma^*_1)}\right)$}
        \label{fig:std_gamma_1}
    \end{subfigure}
    \hfill
    \begin{subfigure}{0.32\textwidth}
        \includegraphics[width=\textwidth]{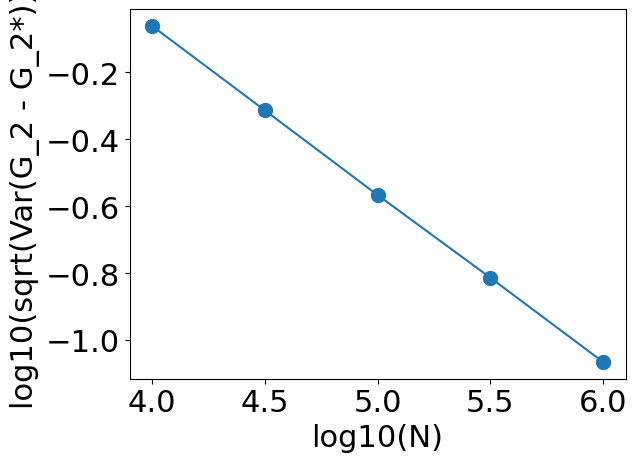}
        \caption{$\log_{10} \left( \sqrt{Var(\gamma_{N,2}-\gamma^*_2)}\right)$}
        \label{fig:std_gamma_2}
    \end{subfigure}
    \hfill
    \begin{subfigure}{0.32\textwidth}
        \includegraphics[width=\textwidth]{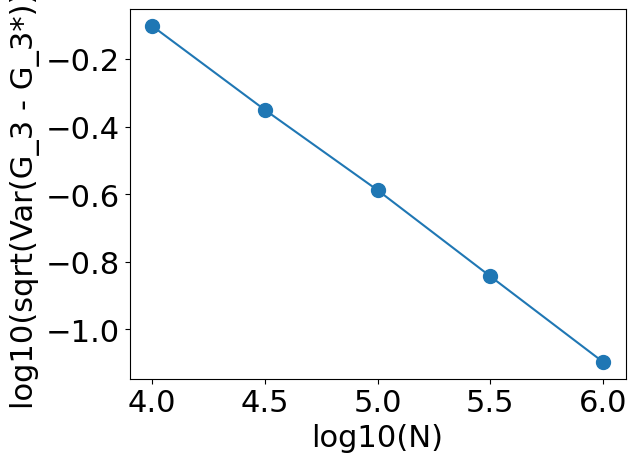}
        \caption{$\log_{10} \left( \sqrt{Var(\gamma_{N,3}-\gamma^*_3)}\right)$}
        \label{fig:std_gamma_3}
    \end{subfigure}
    \begin{subfigure}{0.32\textwidth}
        \includegraphics[width=\textwidth]{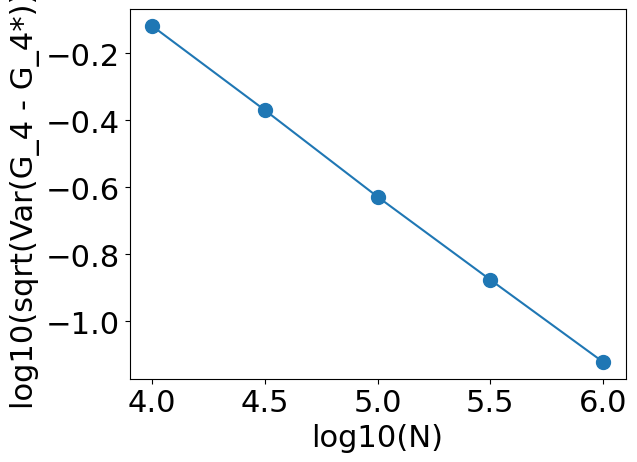}
        \caption{$\log_{10} \left( \sqrt{Var(\gamma_{N,4}-\gamma^*_4)}\right)$}
        \label{fig:std_gamma_4}
    \end{subfigure}
    \begin{subfigure}{0.32\textwidth}
        \includegraphics[width=\textwidth]{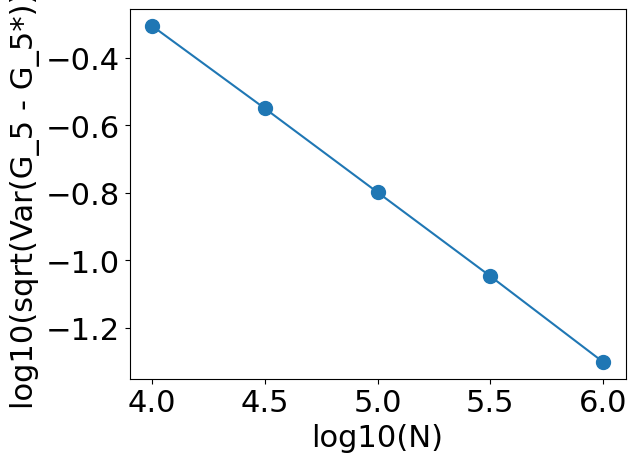}
        \caption{$\log_{10} \left( \sqrt{Var(\gamma_{N,5}-\gamma^*_5)}\right)$}
        \label{fig:std_gamma_5}
    \end{subfigure}

    \caption{Log plot of $\log_{10} \left( \sqrt{Var(\gamma_{N}-\gamma^*)}\right)$ components-wise.}
    \label{fig:gauss_d_2}
\end{figure}

\begin{figure}[H]
    \centering
    \begin{subfigure}{0.32\textwidth}
        \includegraphics[width=\textwidth]{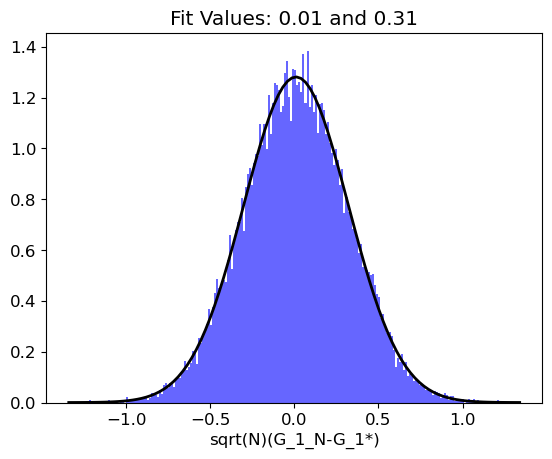}
        \caption{$\sqrt{N}(\gamma_{N,1}-\gamma^*_1)$}
        \label{fig:hist_gamma_1}
    \end{subfigure}
    \hfill
    \begin{subfigure}{0.32\textwidth}
        \includegraphics[width=\textwidth]{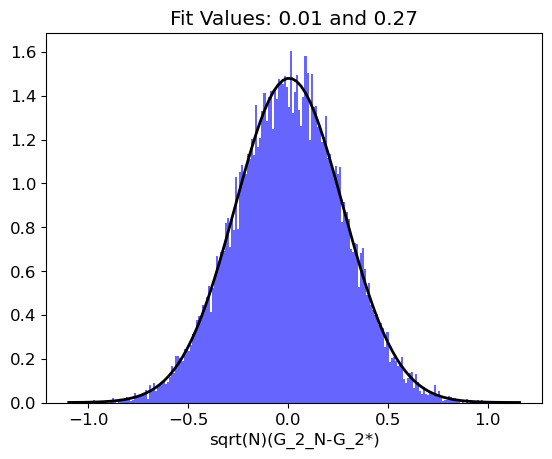}
        \caption{$\sqrt{N}(\gamma_{N,2}-\gamma^*_2)$}
        \label{fig:hist_gamma_2}
    \end{subfigure}
    \hfill
    \begin{subfigure}{0.32\textwidth}
        \includegraphics[width=\textwidth]{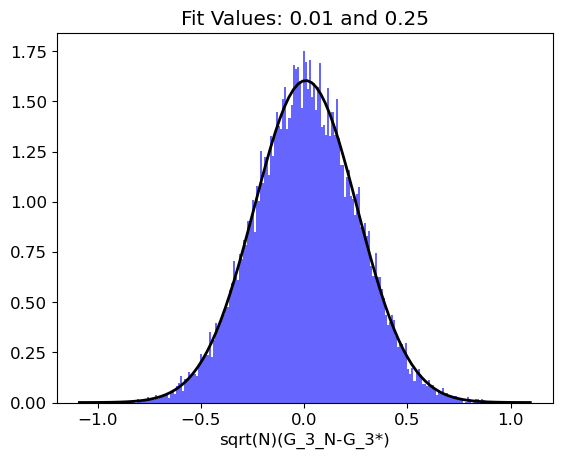}
        \caption{$\sqrt{N}(\gamma_{N,3}-\gamma^*_3)$}
        \label{fig:hist_gamma_3}
    \end{subfigure}
    \begin{subfigure}{0.32\textwidth}
        \includegraphics[width=\textwidth]{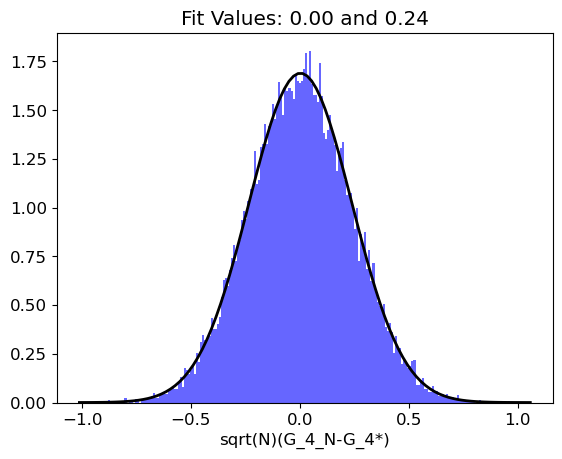}
        \caption{$\sqrt{N}(\gamma_{N,4}-\gamma^*_4)$}
        \label{fig:hist_gamma_4}
    \end{subfigure}
    \begin{subfigure}{0.32\textwidth}
        \includegraphics[width=\textwidth]{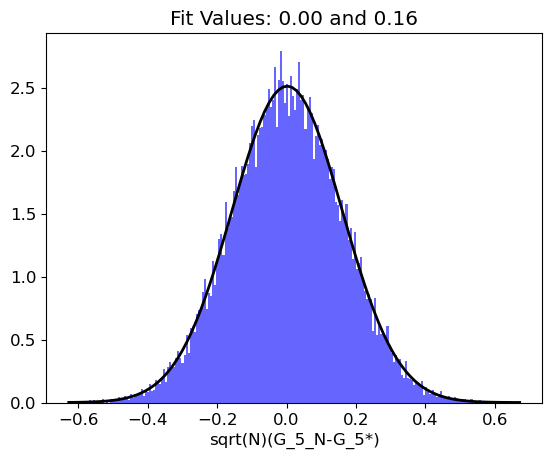}
        \caption{$\sqrt{N}(\gamma_{N,5}-\gamma^*_5)$}
        \label{fig:hist_gamma_5}
    \end{subfigure}

    \caption{Components-wise histograms of $\sqrt{N}(\gamma_{N}-\gamma^*)$ for $N = 10^5$ and plot of a fitted normal distribution.}
    \label{fig:gauss_d_2_bis}
\end{figure}

\begin{figure}[H]
    \centering
    \begin{subfigure}{0.49\textwidth}
        \includegraphics[width=\textwidth]{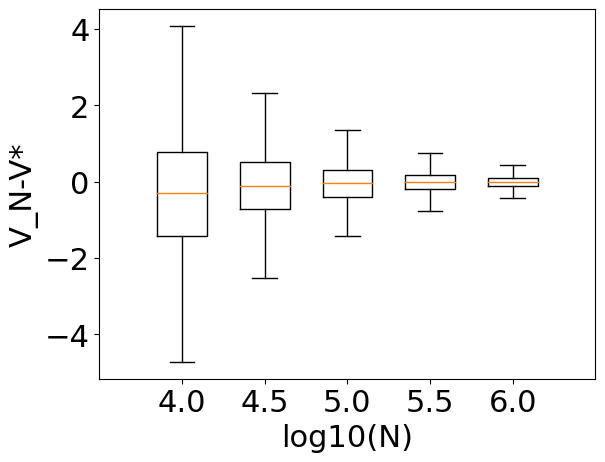}
        \caption{$v_N-v^*$}
        \label{fig:boxplot_vn}
    \end{subfigure}
    \hfill
    \begin{subfigure}{0.49\textwidth}
        \includegraphics[width=\textwidth]{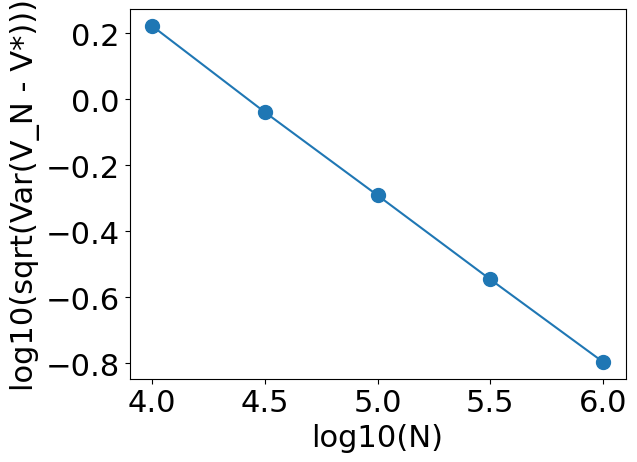}
        \caption{$\log_{10} \left(\sqrt{Var(v_N-v^*)}\right)$}
        \label{fig:std_vn}
    \end{subfigure}

    \caption{Boxplots of $v_N-v^*$ and $\log_{10}$ plot of $\log_{10} \left(\sqrt{Var(v_N-v^*)}\right)$.}
    \label{fig:gauss_d_0}
\end{figure}

\begin{figure}[H]
    \centering
    \includegraphics[width=0.5\linewidth]{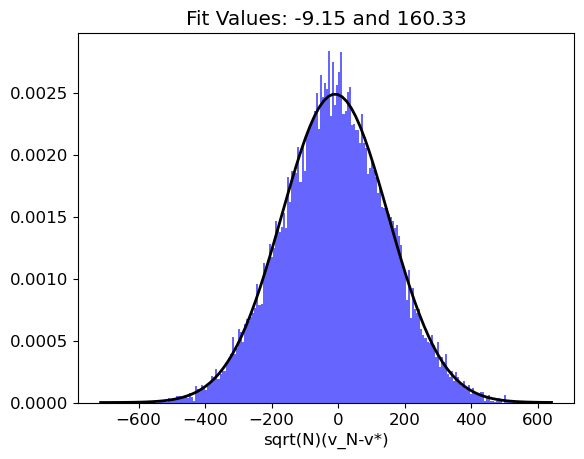}
    \caption{Histogram of $\sqrt{N}(v_N-v^*)$ for $N = 10^5$ and plot of a fitted normal distribution.}
    \label{fig:gauss_d_3}
\end{figure}

\subsubsection{With bounds}
Now, in the case where the optimal parameters reach at least one bound, we cannot compute the theoretical values for $\gamma^*$ and $v^*$.

Looking at the the unconstrained optimal values $(15.113, 12.576, 12.610,  8.702,  3.958)$, from subsection \ref{subsection:gauss_without_bound}, helps us choose the limits $\boldsymbol{\gamma}^{low} = (0,0,15,0,6)$ and $\boldsymbol{\gamma}^{up} = (10,30,30,30,30)$ such that they will be reached. We expect parameter 1 to hit its higher limit and parameters 3 and 5 to hit their lower limits. These claims are confirmed by a simulation with $N=10^6$ where $\boldsymbol{\gamma}_{10^6} = (10,9.747,15,7.867, 6)$ and $v_{10^6} = -48.970$, which is expected more than the unconstrained case $-52.960$.

Figure \ref{fig:gauss_d_4} and \ref{fig:gauss_d_5} show the convergence to a solution $\boldsymbol{\gamma}^*$ which has its components $1,3,5$ bounded (\ref{fig:gamma_n_bounded_1}, \ref{fig:gamma_n_bounded_3}, \ref{fig:gamma_n_bounded_5}) as expected. Figure \ref{fig:gauss_d_6} shows $N^{1/2}(v_N-v^*) \sim \mathcal{N}(0,\sigma^2)$ thanks to Theorem \ref{th:global_rate_cv}.
We approximate $v^*$ by $\mathbb{E}(v_N)$ and $\boldsymbol{\gamma}^*$ by $\mathbb{E}(\boldsymbol{\gamma}_N)$ on k runs. Here we do not have access to a theoretical $\sigma$ value, we estimate it by $\tilde{\sigma} = 226.99$ with the gaussian shape estimation.
Figure \ref{fig:hist_gamma_b_2}  shows an asymptotic Gaussian behavior, while Figure \ref{fig:hist_gamma_b_4} is less well approximated by a Gaussian.

\begin{figure}[H]
    \centering
    \begin{subfigure}{0.3\textwidth}
        \includegraphics[width=\textwidth]{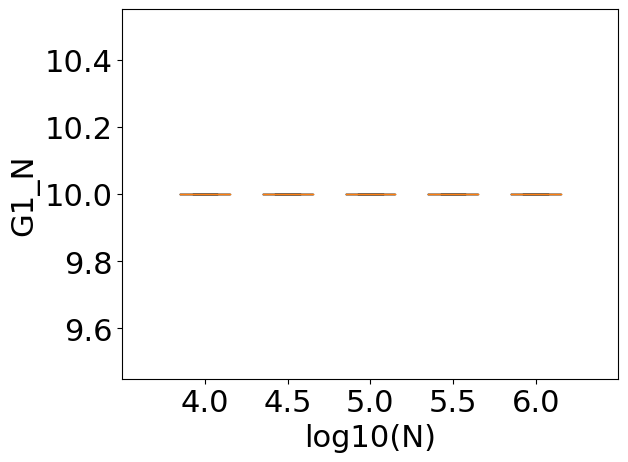}
        \caption{$\boldsymbol{\gamma}_{N,1}$}
        \label{fig:gamma_n_bounded_1}
    \end{subfigure}
    \hfill
    \begin{subfigure}{0.3\textwidth}
        \includegraphics[width=\textwidth]{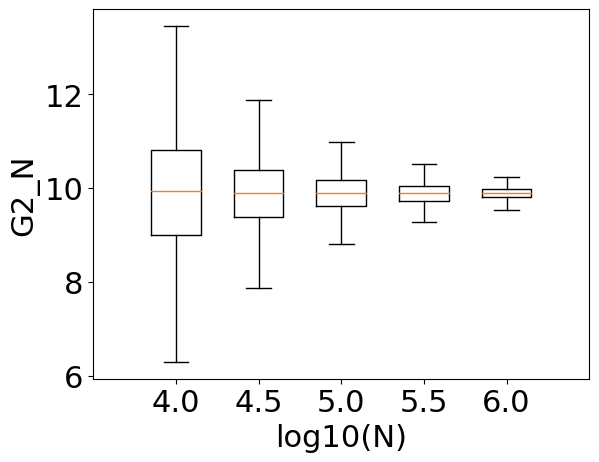}
        \caption{$\boldsymbol{\gamma}_{N,2}$}
        \label{fig:gamma_n_bounded_2}
    \end{subfigure}
    \hfill
    \begin{subfigure}{0.3\textwidth}
        \includegraphics[width=\textwidth]{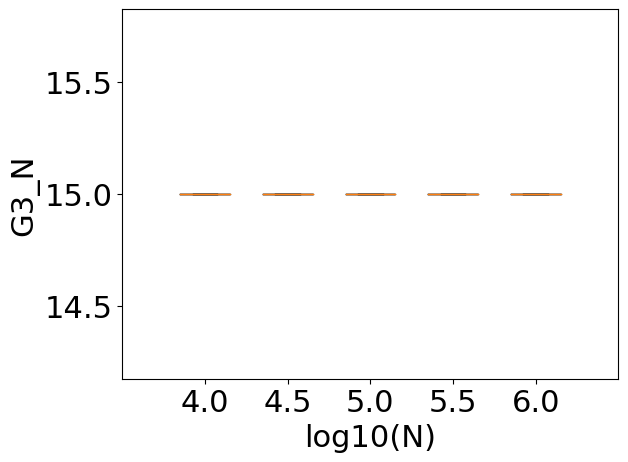}
        \caption{$\boldsymbol{\gamma}_{N,3}$}
        \label{fig:gamma_n_bounded_3}
    \end{subfigure}
    \begin{subfigure}{0.3\textwidth}
        \includegraphics[width=\textwidth]{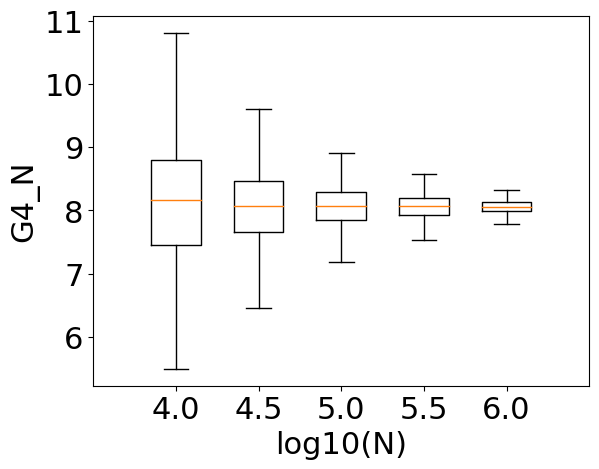}
        \caption{$\boldsymbol{\gamma}_{N,4}$}
        \label{fig:gamma_n_bounded_4}
    \end{subfigure}
    \begin{subfigure}{0.3\textwidth}
        \includegraphics[width=\textwidth]{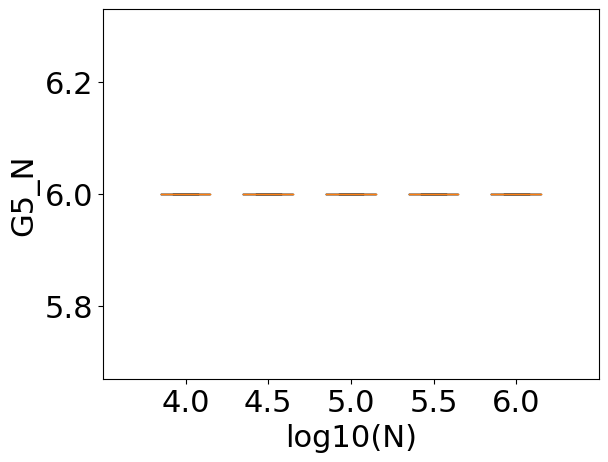}
        \caption{$\boldsymbol{\gamma}_{N,5}$}
        \label{fig:gamma_n_bounded_5}
    \end{subfigure}

    \caption{Component-wise Boxplots of $\boldsymbol{\gamma}_{N}$. (boxplot standard parameters: Q1-1.5IQR,Q1,median,Q3,Q3+1.5IQR)}
    \label{fig:gauss_d_4}
\end{figure}

\begin{figure}[H]
    \centering
    \begin{subfigure}{0.49\textwidth}
        \includegraphics[width=\textwidth]{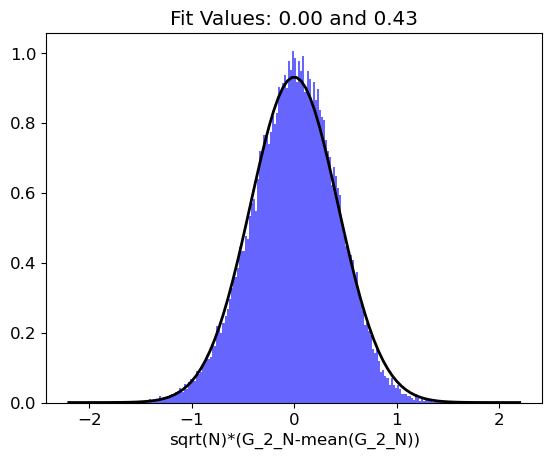}
        \caption{$\sqrt{N}(\gamma_{N,2}-\overline{\gamma_{N,2}})$}
        \label{fig:hist_gamma_b_2}
    \end{subfigure}
    \hfill
    \begin{subfigure}{0.49\textwidth}
        \includegraphics[width=\textwidth]{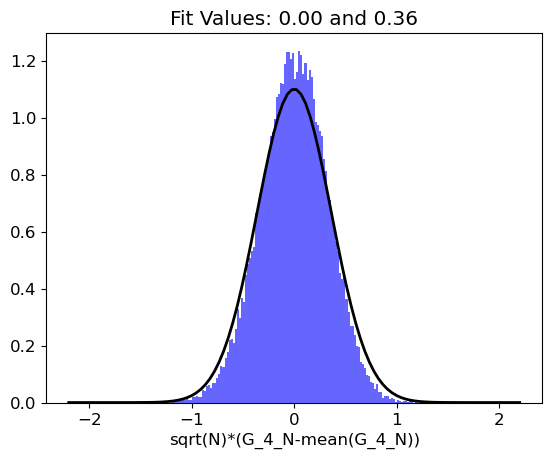}
        \caption{$\sqrt{N}(\gamma_{N,4}-\overline{\gamma_{N,4}})$}
        \label{fig:hist_gamma_b_4}
    \end{subfigure}

    \caption{Histograms of $\sqrt{N}(\gamma_{N}-\gamma^*)$ for components 2 and 4, for $N = 10^5$ and plot of a fitted normal distribution on each.}
    \label{fig:gauss_d_4_bis}
\end{figure}

\begin{figure}[H]
    \centering
    \includegraphics[width=0.5\linewidth]{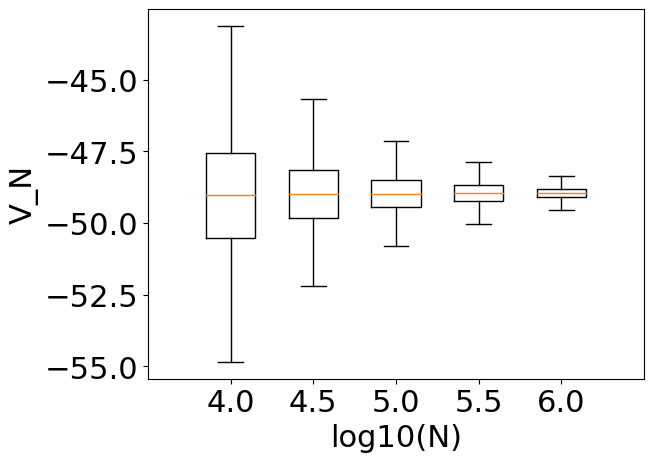}
    \caption{Boxplots of $v_N$. (boxplot standard parameters: Q1-1.5IQR,Q1,median,Q3,Q3+1.5IQR)}
    \label{fig:gauss_d_5}
\end{figure}

\begin{figure}[H]
    \centering
    \includegraphics[width=0.5\linewidth]{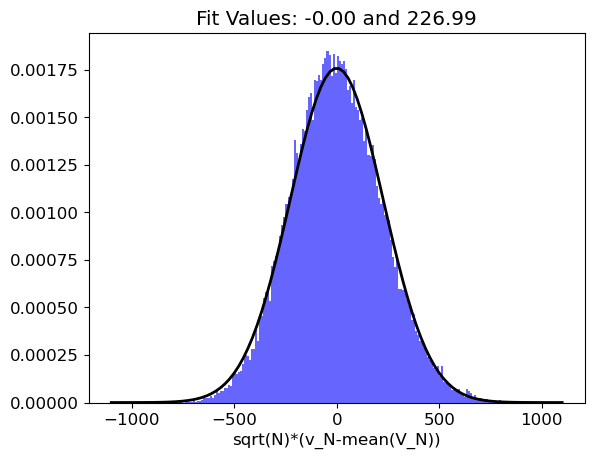}
    \caption{Histogram of $\sqrt{N}(v_N-\Bar{v_N})$ for $N = 10^5$ and plot of a fitted normal distribution.}
    \label{fig:gauss_d_6}
\end{figure}

\subsection{Multi laws setting}

In practice, we only have access to a sample of data for each Line of Business (LoB) returns. This sample may come from complex modeling and may not fit any known law in practice. All LoBs have different distributions with very different risk profiles.

We choose to have 5 LoBs and for each, we choose a distribution and shifted it to obtain a mean of 1 like in the Gaussian case. Let us consider $\Tilde{\textbf{X}}_i \ ,  i=1,\dots,5$, five random variables. $\Tilde{\textbf{X}}_1, \Tilde{\textbf{X}}_2$ follow a generalized Pareto distribution with $\alpha_{pareto}$ parameter $0.45$ and $0.25$ (as it is defined on the scipy python package), $\Tilde{\textbf{X}}_3, \Tilde{\textbf{X}}_4$ follow a log normal distribution with parameters $(\mu, \sigma)=(1.7, 1)$ and $(\mu, \sigma)=(1.3,1)$. Finally, $\Tilde{\textbf{X}}_5$ follows a normal distribution with $(\mu, \sigma) = (2,6)$.

Let us define $\textbf{X}_i = 1 - (\Tilde{\textbf{X}}_i - \Bar{\Tilde{\textbf{X}}}_i)$ for $i=1,\dots,5$. In such a setting, the standalone $VaR$ and $CVaR$ can be calculated analytically.

\begin{table}[h!]
    \centering
    \begin{tabular}{||c c c c c||}
        \hline
        i & $\mu(\textbf{X}_i)$ & $\sigma(\textbf{X}_i)$ & $VaR_{0.99}(-\textbf{X}_i)$ & $CVaR_{0.99}(-\textbf{X}_i)$ \\ [0.5ex]
        \hline\hline
        1 & 1                   & 5.75                   & 12.61                       & 27.05                        \\
        2 & 1                   & 1.89                   & 6.32                        & 10.53                        \\
        3 & 1                   & 17.49                  & 46.94                       & 80.55                        \\
        4 & 1                   & 4.89                   & 17.25                       & 53.73                        \\
        5 & 1                   & 6                      & 12.96                       & 15.99                        \\ [1ex]
        \hline
    \end{tabular}
    \caption{distribution description for the multi laws case}
    \label{table:mixed_case_statistic}
\end{table}

We choose three different dependency settings: a Gaussian copula, a clayton copula and the independent case. For the clayton case we use $\alpha_{clayton}=2$ and in the Gaussian case we use the following correlation matrix picked at random.

\begin{equation*}
    \Sigma =
    \begin{pmatrix}
        1       & -0.1285 & 0.3979  & -0.4731 & 0.3879  \\
        -0.1285 & 1       & -0.0574 & -0.2253 & -0.3532 \\
        0.3979  & -0.0574 & 1       & -0.5363 & 0.12    \\
        -0.4731 & -0.2253 & -0.5363 & 1       & 0.0999  \\
        0.3879  & -0.3532 & 0.12    & 0.0999  & 1       \\
    \end{pmatrix}.
\end{equation*}

We run the optimization with the following parameters $K=100$, without bounds and with bounds $\boldsymbol{\gamma}^{low} = (2,0,0,1,0)$ and $\boldsymbol{\gamma}^{up} = (10,5,10,10,10)$ and $N=10^6$.

\begin{table}[h!]
    \centering
    \begin{tabular}{||c c c c c c | c||}
        \hline
        $N=10^6$    & $\gamma_1$ & $\gamma_2$ & $\gamma_3$ & $\gamma_4$ & $\gamma_5$ & $v$     \\ [0.5ex]
        \hline\hline
        \hline
        without bound                                                                          \\ [0.5ex]
        \hline\hline
        independent & 1.407      & 7.596      & 0.101      & 1.196      & 4.571      & -14.922 \\
        Gaussian    & 0.575      & 10.386     & 0.0628     & 1.131      & 5.498      & -17.655 \\
        clayton     & 0.7521     & 6.535      & 0.0178     & 0.656      & 3.794      & -11.763 \\
        [1ex]
        \hline \hline
        with bounds                                                                            \\ [0.5ex]
        \hline\hline
        independent & 2          & 5          & 0.146      & 1.501      & 5.343      & -13.987 \\
        Gaussian    & 2          & 5          & 0.0333     & 1.577      & 4.839      & -13.404 \\
        clayton     & 2          & 5          & 2.968      & 1          & 2.969      & -10.959 \\
        [1ex]
        \hline
    \end{tabular}
    \caption{multi laws results}
    \label{table:multi_laws_results}
\end{table}

At optimum and for all cases, we expect that the LoBs with the lowest $CVaR$, like 2 and 5 are prioritized over those with the highest $CVaR$, like 3. The case without bounds leads us to expect that LoB 1 and 2 will reach their bounds, and that is what we are seeing. Also, the bounded cases, as expected, slow the diversification effects and then the return are worse than without.

In the independent case, $\gamma_i$ is approximately inversely proportional to $CVaR_{0.99}(-\textbf{X}_i)$, except for the bounded LoBs. The Gaussian dependency case without bound gives a better return thanks to a diversification with the chosen correlation matrix, like LoB 2 with LoBs 4 and 5 because of $\Sigma_{2,4} = -0.2253, \Sigma_{2,5} = -0.3532$. However, this diversification benefit cannot be exploited in the bounded case, mainly due to LoB 2 which is limited to 5. The Clayton case intensifies the queue positive dependency and weakens the $CVaR$ diversification, so the return is worse.

\section{Appendix}
There are two parts to this appendix, the first is the proof of the Equivalent Theorem used in \ref{subse:eq_setting} and the second part is the proof of the optimal value formula for the Gaussian case in \ref{subsection:Gaussian} when the bounds are not reached.

\subsection{Equivalent Theorem proof}
We use in section \ref{subse:eq_setting} a result from
Krokhmal P., Jonas Palmquist J., Uryasev S. (2002) \cite{Uryasev2002:9} on the equivalence of two optimization problems but this particular property does not appear explicitly in the article so we give a proof here.

The proof is based on the Karush-Kuhn-Tucker necessary and sufficient conditions.

\equivalenceprop*

\begin{proof}
    Let us write down the necessary and sufficient Karush-Kuhn-Tucker conditions for problems \eqref{eq:Reward_True} and \eqref{eq:Reward_True_F}.

    For $R(.,C_\alpha(.))$ to achive its minimum at $\gamma^* \in \Gamma$, it must exist a constant $\mu_1$ such that, for all $\gamma \in \Gamma$:
    \begin{align*}\tag{KKT-\eqref{eq:Reward_True}}\label{eq:KKT-True}
        -R(\gamma^*, C_\alpha(\gamma^*)) + \mu_1 C_\alpha(\gamma^*) & \leq -R(\gamma, C_\alpha(\gamma)) + \mu_1 C_\alpha(\gamma), \\
        \mu_1 (C_\alpha(\gamma^*)-K)                                & = 0. \ \
    \end{align*}
    If $\mu_1 \geq 0$, then the conditions are also sufficient.

    We have the same kind of conditions for \eqref{eq:Reward_True_F}. For $(\gamma,\zeta) \in \mathcal{U}$:
    \begin{align*}\tag{KKT-\eqref{eq:Reward_True_F}}\label{eq:KKT-eq}
        -R(\gamma^*, g(\gamma^*,\zeta^*)) + \mu_2 g(\gamma^*,\zeta^*) & \leq -R(\gamma, g(\gamma,\zeta)) + \mu_2 g(\gamma,\zeta), \\
        \mu_2 (g(\gamma^*,\zeta^*)-K)                                 & = 0, \ \  \mu_2 \geq 0.
    \end{align*}

    First, suppose that $\gamma^*$ is a solution to \eqref{eq:Reward_True} and $\zeta^* \in A_\alpha(\gamma^*)$. Let us show that $(\gamma^*,\zeta^*)$ is a solution to \eqref{eq:Reward_True_F}. Using necessary and sufficient conditions \eqref{eq:KKT-True}
    \begin{align*}
         & -R(\gamma^*, g(\gamma^*,\zeta^*)) + \mu_1 g(\gamma^*,\zeta^*) = -R(\gamma^*, C_\alpha(\gamma^*)) + \mu_2 C_\alpha(\gamma^*)                                          \\
         & \leq -R(\gamma, C_\alpha(\gamma)) + \mu_1 C_\alpha(\gamma) = -R(\gamma,\underset{\zeta}{\min} \ g(\gamma,\zeta)) + \mu_1 \underset{\zeta}{\min} \  g(\gamma,\zeta), \\
         & \leq -R(\gamma, g(\gamma,\zeta)) + \mu_1 g(\gamma,\zeta) \ \text{because} \ R(\boldsymbol{\gamma}, . ) \ \text{is not increasing},                                   \\
    \end{align*}
    and
    \begin{align*}
         & \mu_2 (g(\gamma^*,\zeta^*)-K) = \mu_1 (C_\alpha(\gamma^*)-K) 0, \ \  \mu_1 \geq 0, (\gamma,\zeta) \in \mathcal{U}.
    \end{align*}
    Thus, \eqref{eq:KKT-eq} conditions are satisfied and $(\gamma^*,\zeta^*)$ is a solution to \eqref{eq:Reward_True_F}.

    Now let us suppose that $(\gamma^*,\zeta^*)$ achieves the minimum of \eqref{eq:Reward_True_F} and $\mu_2 > 0$. For a fixed $\gamma^*$, $\zeta^*$ minimizes the function $\zeta \mapsto -R(\gamma^*, g(\gamma^*,\zeta)) + \mu_2 g(\gamma^*,\zeta)$, and, consequently, the function $\zeta \mapsto g(\gamma^*,\zeta)$ because $R(\boldsymbol{\gamma}, . )$ is not increasing. Then, it implies that $\zeta^* \in A_\alpha(\gamma^*)$. Thus
    \begin{align*}
         & -R(\gamma^*, C_\alpha(\gamma^*)) + \mu_2 C_\alpha(\gamma^*) = -R(\gamma^*, g(\gamma^*,\zeta^*)) + \mu_2 g(\gamma^*,\zeta^*)            \\
         & \leq -R(\gamma, g(\gamma,V_\alpha(\gamma))) + \mu_2 g(\gamma,V_\alpha(\gamma)) =  -R(\gamma, C_\alpha(\gamma)) + \mu_2 C_\alpha(\gamma)
    \end{align*}
    and
    \begin{align*}
         & \mu_2 (C_\alpha(\gamma^*)-K) = \mu_2 (g(\gamma^*,\zeta^*)-K) = 0, \ \ \mu_2 \geq 0, \ \gamma \in \Gamma.
    \end{align*}
    We proved that \eqref{eq:KKT-eq} are satisfied, i.e. $\gamma^*$ is a solution to \eqref{eq:Reward_True} which completes the proof.
\end{proof}

\subsection{Proof of Proposition ~\ref{prop:strict_convexity_CVaR}}
In this appendix, we introduce two lemmas and provide the proof of Proposition~\ref{prop:strict_convexity_CVaR}.

\begin{lemma}\label{th:lemma_comonotonicity}
    Let $\alpha \in ]0,1[$ and $Y$ (resp. $Z$) be a real valued random variable with a density on $\mathbb{R}$ and cumulative distribution function $F_Y$ (resp. $F_Z$). We define $A_Y := \vmathbb{1}_{\{ U_Y \geq \alpha \}}$ with $U_Y := F_Y(Y)$ (resp. $A_Z := \vmathbb{1}_{\{ U_Z \geq \alpha \}}$ with $U_Z := F_Z(Z)$). If $\mathbb{E}(Y A_Y) + \mathbb{E}(Z A_Z) -  \mathbb{E}((Y+Z)A_{Y+Z}) = 0$, then $\{ U_Y \geq \alpha \}  = \{ U_{Y+Z} \geq \alpha \} = \{ U_Z \geq \alpha \}$.
\end{lemma}
\begin{proof}
    Note that $A_Y$ and $A_{Y+Z}$ have the same distribution thus $\mathbb{E}(A_Y - A_{Y+Z}) = 0$. Then, for all $m \in \mathbb{R}$ we have $\mathbb{E}(Y (A_Y-A_{Y+Z})) = \mathbb{E}((Y-m) (A_Y-A_{Y+Z}))$. Take $m=F^{-1}_Y(\alpha)$, with $F^{-1}$ the cdf general inverse; we have $(Y-m)(A_Y-A_{Y+Z})\geq 0$.
    As a consequence $\mathbb{E}(Y (A_Y-A_{Y+Z})) \geq 0$ and $\mathbb{E}(Z (A_Z-A_{Y+Z})) \geq 0$. The assumption can be rewritten as $\mathbb{E}(Y (A_Y-A_{Y+Z})) + \mathbb{E}(Z (A_Z-A_{Y+Z})) = 0$. We conclude that $ \mathbb{E}(Y (A_Y-A_{Y+Z})) = \mathbb{E}(Z (A_Z-A_{Y+Z})) = 0$. In addition $\mathbb{E}((Y-m) (A_Y-A_{Y+Z}))=0$ if and only if $A_Y = A_{Y+Z}$ a.s, similarly $A_Z = A_{Y+Z}$ a.s. because $(Y-m) \neq 0$ a.s. on $\mathbb{R}$. Finally, $\{ U_Y \geq \alpha \}  = \{ U_{Y+Z} \geq \alpha \} = \{ U_Z \geq \alpha \}$.
\end{proof}

\begin{lemma}\label{th:lemma_affine_form}
    let $\mathcal{R}$ be a d-dimensional connected subset of $\mathbb{R}^d$, $\phi_1, \phi_2$ two non-zero $\mathbb{R}^d$-affine forms. If $\phi_1 \phi_2 \geq 0$ on $\mathcal{R}$ and $\mathring{\mathcal{R}} \cap \ker(\phi_i) \neq \emptyset, \ i =1,2$ then $\phi_1 = a \phi_2$ with $a >0$.
\end{lemma}
\begin{proof}
    Suppose $\ker(\phi_1) \neq \ker(\phi_2)$ which are two (d-1)-dimensional affine hyperplanes. There exists $x \in \mathring{\mathcal{R}} \cap (\ker(\phi_1) \setminus \ker(\phi_2)) $, so $\phi_1(x)=0$ and $\phi_2(x) \neq 0$. Let us take the case where $\phi_2(x) < 0$, there exists a neighbourhood $v(x) \subset \mathcal{R}$, such that $ \forall y \in v(x) \ \phi_2(y) < 0$. $v(x)$ is an open set in $\mathbb{R}^d$ and $\ker(\phi_1)$ is (d-1)-dimensional, thus there exists $y'\in v(x)$, such that $\phi_1(y')>0$ and therefore $\phi_1(y')\phi_2(y')<0$ with $y'\in \mathcal{R}$. Symmetrically, it's the same with $\phi_2(x) > 0$. Finally $\ker(\phi_1) = \ker(\phi_2)$ and so $\phi_1(x) = a\phi_2(x)$ with $a >0$.
\end{proof}

\begin{proof}[Proof of Proposition~\ref{prop:strict_convexity_CVaR}]

    The map $\boldsymbol{\gamma} \mapsto C_\alpha(\boldsymbol{\gamma})$ is convex on $\mathbb{R}^d_*$ as a consequence of the sub-additivity property of the Conditional Value-at-Risk. We now aim to prove that this mapping is strictly convex. To this end, let $\boldsymbol{\gamma}_1, \boldsymbol{\gamma}_2 \in \mathbb{R}^d_*$ be such that, for all $\mu \in (0,1)$, the following equality holds:
    \begin{equation}
        C_\alpha(\mu \boldsymbol{\gamma}_1) + C_\alpha((1 - \mu)\boldsymbol{\gamma}_2) = C_\alpha(\mu \boldsymbol{\gamma}_1 + (1 - \mu)\boldsymbol{\gamma}_2). \label{eq:uniticity_tvar}
    \end{equation}

    We use the notations from Lemma~\ref{th:lemma_comonotonicity} and recall that
    \begin{equation*}
        C_\alpha(\boldsymbol{\gamma}) = \frac{1}{1 - \alpha} \, \mathbb{E}\left[ -\boldsymbol{\gamma}^T \mathbf{X} \cdot \vmathbb{1}_{\{ U_{-\boldsymbol{\gamma}^T \mathbf{X}} \geq \alpha \}} \right].
    \end{equation*}
    Let $Y := -\mu\, \boldsymbol{\gamma}_1^T \mathbf{X}$ and $Z := -(1 - \mu)\, \boldsymbol{\gamma}_2^T \mathbf{X}$. Then equality~\eqref{eq:uniticity_tvar} rewrites as
    \begin{equation*}
        \mathbb{E}(Y A_Y) + \mathbb{E}(Z A_Z) = \mathbb{E}((Y + Z) A_{Y+Z}),
    \end{equation*}
    so by Lemma~\ref{th:lemma_comonotonicity} and (P3), we have
    \begin{equation*}
        \{ U_Y \geq \alpha \} = \{ U_Z \geq \alpha \} = \{ U_{Y+Z} \geq \alpha \}.
    \end{equation*}
    These events can be rewritten as
    \begin{equation*}
        \{ U_Y \geq \alpha \} = \{ Y \geq \text{VaR}_\alpha(Y) \} = \{ -\boldsymbol{\gamma}_1^T \mathbf{X} \geq t_1 \}, \quad \text{with } t_1 := \text{VaR}_\alpha(-\boldsymbol{\gamma}_1^T \mathbf{X}),
    \end{equation*}

    and similarly for $Z$ with threshold $t_2$. Therefore,
    \begin{equation*}
        \mathbb{P} \left( (-\boldsymbol{\gamma}_1^T \mathbf{X} - t_1)( -\boldsymbol{\gamma}_2^T \mathbf{X} - t_2 ) \geq 0 \right) = 1.
    \end{equation*}

    Define $\phi_i(\mathbf{x}) := -\boldsymbol{\gamma}_i^T \mathbf{x} - t_i$ for $i = 1,2$. Then $\phi_1 \cdot \phi_2 \geq 0$ almost surely on $\mathcal{R}_\mathbf{X}$. Since $t_i \in \text{supp}(-\boldsymbol{\gamma}_i^T \mathbf{X})$, the level sets $\ker(\phi_i)$ intersect the interior of the support. By Lemma~\ref{th:lemma_affine_form}, this implies $\phi_1 = a\, \phi_2$ for some $a > 0$, hence $\boldsymbol{\gamma}_1 = a\, \boldsymbol{\gamma}_2$.

    Finally, the positive homogeneity of $C_\alpha$ and equality in~\eqref{eq:uniticity_tvar} yield $a = 1$, so $\boldsymbol{\gamma}_1 = \boldsymbol{\gamma}_2$.

\end{proof}

\subsection{The Gaussian case resolution}
We use in section \ref{subsection:Gaussian} some theoretical results, we give a proof here.

The Lagrangian of \eqref{eq:P} in the Gaussian case can be written as:
\begin{equation}\label{eq:gauss_lagrange_def}
    \begin{aligned}
        \mathcal{L}(\boldsymbol{\gamma}, (\lambda,\Bar{\boldsymbol{\mu}}, \underline{\boldsymbol{\mu}})) := & - \boldsymbol{\gamma}^T \textbf{1}_d + \lambda(- \boldsymbol{\gamma}^T \textbf{1}_d+ T_Z\sqrt{\boldsymbol{\gamma}^T \Sigma \boldsymbol{\gamma}}-K)   \\
                                                                                                            & +\Bar{\boldsymbol{\mu}}(\boldsymbol{\gamma}-\boldsymbol{\gamma}^{up}) - \underline{\boldsymbol{\mu}}(\boldsymbol{\gamma}-\boldsymbol{\gamma}^{low}).
    \end{aligned}
\end{equation}

$\frac{\partial \mathcal{L}}{\partial \lambda} = 0 $ implies
\begin{equation}\label{eq:gauss_lagrange_1}
    K = - \boldsymbol{\gamma}^T \textbf{1}_d+ T_Z\sqrt{\boldsymbol{\gamma}^T \Sigma \boldsymbol{\gamma}},
\end{equation}
and $\frac{\partial \mathcal{L}}{\partial \boldsymbol{\gamma}} = 0$ implies
\begin{equation}\label{eq:gauss_lagrange_2}
    \left(\frac{1-(\Bar{\boldsymbol{\mu}}-\underline{\boldsymbol{\mu}})}{\lambda} + 1\right) \textbf{1}_d =  T_Z \frac{\Sigma \boldsymbol{\gamma}}{\sqrt{\boldsymbol{\gamma}^T \Sigma \boldsymbol{\gamma}}}.
\end{equation}

If the optimal solution is not constrained by the bounds, then $\Bar{\boldsymbol{\mu}}=\underline{\boldsymbol{\mu}}=0$. Denote $y = \frac{1}{\lambda} + 1$ and $\sigma_S = \sqrt{\boldsymbol{\gamma}^T \Sigma \boldsymbol{\gamma}}$

We multiply \eqref{eq:gauss_lagrange_2} on the left by $\boldsymbol{\gamma}$,
\begin{equation*}
    \frac{y}{T_Z} \boldsymbol{\gamma}^T \textbf{1}_d = \frac{\boldsymbol{\gamma}^T \Sigma \boldsymbol{\gamma}}{\sqrt{\boldsymbol{\gamma}^T \Sigma \boldsymbol{\gamma}}},
\end{equation*}
so
\begin{equation}\label{eq:gauss_lagrange_3}
    \frac{y}{T_Z} \boldsymbol{\gamma}^T \textbf{1}_d = \sigma_S.
\end{equation}

We multiply \eqref{eq:gauss_lagrange_2} on the left by $\textbf{1}_d^T \Sigma^{-1}$,

\begin{equation*}
    y (\textbf{1}_d^T \Sigma^{-1} \textbf{1}_d) = \frac{T_Z}{\sigma_S}\textbf{1}_d^T \boldsymbol{\gamma}.
\end{equation*}
Now we replace  $\textbf{1}_d^T \boldsymbol{\gamma}$ by $\frac{T_Z}{y}\sigma_S$ with \eqref{eq:gauss_lagrange_3} and isolate $y$,
\begin{equation}\label{eq:gauss_lagrange_4}
    y = \frac{T_Z}{\sqrt{\textbf{1}_d^T \Sigma^{-1} \textbf{1}_d}}.
\end{equation}

In the same way, we replace $\textbf{1}_d^T \boldsymbol{\gamma}$ by $\frac{T_Z}{y}\sigma_S$ with \eqref{eq:gauss_lagrange_3} in \eqref{eq:gauss_lagrange_1},
\begin{equation*}
    K = -\frac{T_Z}{y} \sigma_S + T_Z \sigma_S,
\end{equation*}
and replace $\frac{T_Z}{y}$ by $\sqrt{\textbf{1}_d^T \Sigma^{-1} \textbf{1}_d}$ with \eqref{eq:gauss_lagrange_4} and isolate $\sigma_S$,
\begin{equation}\label{eq:gauss_lagrange_5}
    \sigma_S = K  \left( \frac{1}{T_Z-\sqrt{\textbf{1}_d^T \Sigma^{-1} \textbf{1}_d}} \right).
\end{equation}

Let us now replace $y$ and $\sigma_S$ in \eqref{eq:gauss_lagrange_2},
\begin{align*}
     & \frac{T_Z}{\sqrt{\textbf{1}_d^T \Sigma^{-1} \textbf{1}_d}} \textbf{1}_d = T_Z \frac{\Sigma \boldsymbol{\gamma} }{K \left( \frac{1}{T_Z-\sqrt{\textbf{1}_d^T \Sigma^{-1} \textbf{1}_d}} \right)}, \\
\end{align*}
finally
\begin{equation}
    \boldsymbol{\gamma} = K \left( \frac{1}{T_Z-\sqrt{\textbf{1}_d^T \Sigma^{-1} \textbf{1}_d}} \right) \frac{\Sigma^{-1} \textbf{1}_d }{\sqrt{\textbf{1}_d^T \Sigma^{-1} \textbf{1}_d}}.
\end{equation}

\printbibliography

\end{document}